\documentclass{CSML}
\pdfoutput=1

\usepackage{lastpage}
\newcommand{\dOi}{13(3:11)2017}
\lmcsheading{}{1--\pageref{LastPage}}{}{}%
{Nov.~08, 2016}{Aug.~15, 2017}{}


\usepackage{hyperref}
\hypersetup{hidelinks}
\usepackage{amssymb}
\usepackage{amsmath}
\usepackage{amsthm}
\usepackage[caption=false,font=footnotesize,captionskip=14pt]{subfig}
\usepackage{wrapfloat}
\usepackage{color}

\newif\ifAppendixOn\AppendixOntrue

\newcommand{\set}[1]{\{\, #1 \,\}}

\newcommand{\X}[0]{{\mathcal X}}
\newcommand{\Y}[0]{{\mathcal Y}}
\newcommand{\Z}[0]{{\mathcal Z}}
\newcommand{\W}[0]{{\mathcal W}}
\renewcommand{\S}[0]{{\mathcal S}}
\newcommand{\I}[0]{{\mathcal I}}

\newcommand{\compd}{\times}
\newcommand{\compn}{\mathbin{\parallel}}

\newcommand{\pid}{\pi^{\dagger}}

\newcommand{\gd}{g^{\dagger}}

\newcommand{\Mmin}[1]{M^{\min}_{#1}}
\newcommand{\Mmax}[1]{M^{\max}_{#1}}
\newcommand{\Mminpi}{M^{\min}_{g,\pi}}
\newcommand{\Mmaxpi}{M^{\max}_{g,\pi}}
\newcommand{\MminInfpi}{M^{\min}_{\infty,\pi}}
\newcommand{\MmaxInfpi}{M^{\max}_{\infty,\pi}}

\newcommand{\Hinf}{H_{\infty}}
\newcommand{\Iinf}{I_{\infty}}

\newcommand{\HMin}{H^{\min}}

\newcommand{\Hg}{H_{g}}
\newcommand{\Ig}{I_{g}}
\newcommand{\IgA}{I_{g_1}}
\newcommand{\IgB}{I_{g_2}}

\newcommand{\HgMin}{H_g^{\min}}

\newcommand{\Cg}{C_{\!g}}

\newcommand{\Cinf}{C_{\infty}}

\newcommand{\xor}{\oplus}

\newcommand{\leakiEst}{\textsf{leakiEst}}

\begin{document}

\title[On the Compositionality of Quantitative Information Flow]{On the
  Compositionality of \\ Quantitative Information Flow}

\author{Yusuke Kawamoto}	
\address{AIST, Japan}	

\author{Konstantinos Chatzikokolakis}	
\address{CNRS and LIX, \'{E}cole Polytechnique, France}

\author{Catuscia Palamidessi}	
\address{INRIA Saclay and LIX, \'{E}cole Polytechnique, France}	



\titlecomment{A preliminary version of this work, without proofs, appeared in~\cite{KawamotoCP14qest}.}


\begin{abstract}
Information flow is the branch of security that studies the leakage of information due to correlation between secrets and observables.
Since in general such correlation cannot be avoided completely, it is important to quantify the leakage.
The most followed approaches to defining appropriate measures are those based on information theory.
In particular, one of the most successful approaches is the recently proposed $g$-leakage framework, which encompasses most of the information-theoretic ones.
A problem with $g$-leakage, however, is that it is
defined in terms of a minimization problem, which, in the case of large systems, can be computationally rather heavy. 
In this paper we study the case in which the channel associated to the system can be decomposed into simpler channels, 
which typically happens when the  observables consist of multiple components. 
Our main contribution is the derivation of bounds on the (multiplicative version of) $g$-leakage of the whole system in terms of the $g$-leakages of its components. 
We also consider  the particular cases of min-entropy leakage and of parallel channels, generalizing and systematizing results from the literature. 
We demonstrate the effectiveness of our method and evaluate the precision of our bounds using examples.
\end{abstract}

\maketitle

%
\section{Introduction}

The problem of preventing confidential information from being leaked is a fundamental concern in 
the modern society, where the pervasive use of automatized devices makes it hard to predict and control the {\em information flow}.
For instance, the output of a program might reveal information about a sensitive variable, or side channel information might reveal the secret key stored in a smart card.
While early research focused on trying to  achieve \emph{non-interference} (i.e., no leakage), it is nowadays recognized that, in practical situations, some amount of leakage  is unavoidable. 
Therefore an active area of research on information flow is dedicated to the development of theories to 
\emph{quantify}  the amount of leakage, and of methods to  minimize it. See, for instance, \cite{Clark:01:QAPL,Malacaria:07:POPL,Koepf:07:CCS,Chatzikokolakis:08:IC,Chatzikokolakis:08:JCS,Boreale09:iandc,Smith:09:FOSSACS,Boreale:11:FOSSACS,Alvim0KP17corr}.

Among these theories, min-entropy leakage \cite{Smith:09:FOSSACS,Braun:09:MFPS} 
has become quite popular, partly due to its clear operational interpretation in terms of  one-try attacks.
This  quite basic  setting has been recently extended 
to the \emph{$g$-leakage} framework \cite{Alvim:12:CSF}. The main novelty consists in the introduction of  gain functions, 
that permit  to quantify the vulnerability of a secret in terms of the gain of the adversary, thus   
allowing to model a wide variety of operational scenarios and different adversaries
\footnote{In fact, all convex functions can be expressed in terms of $g$-vulnerability~\cite{AlvimCMMPS16csf}.}.
For instance, min-entropy, Shannon entropy and guessing entropy can all be seen as instances of $g$-vulnerability~\cite{Alvim:12:CSF,AlvimCMMPS14csf}.

While $g$-leakage  is appealing for its generality and flexible operational interpretation, 
its computation is not trivial. Like most of the quantitative approaches, its definition is based on the probabilistic correlation between the secrets and the observables. 
Such correlation is usually expressed in terms of an \emph{information-theoretic channel}, 
where the secrets and the observables respectively constitute the input and the output. 
Note that the channel abstracts from the specific model describing the system.
A program is usually expressed in a programming language while a protocol might be better described by an automaton.
Still, in both cases, the leakage only depends on the conditional probabilities of each output given input, i.e., on the \emph{channel matrix}.
The computation of the channel matrix from the system can be performed via model checking 
(see, e.g., \cite{Andres:10:TACAS}), if a system is completely specified and it is not too complicated. 
Once the matrix is known, the computation of the $g$-leakage involves solving an optimization problem. 
In fact, computing leakage involves finding the ``guess'' that optimizes the expected gain of the adversary~\cite{Alvim:12:CSF}.
This can be quite costly when the matrix is large. 

Worse yet, in many cases it is not possible to compute the channel matrix exactly, for instance  because the system may be too complicated, 
or because the conditional probabilities are partially determined by unknown factors. Fortunately,  there are statistical methods that allow to approximate 
the channel matrix and the leakage~\cite{chatzikokolakis10,ChothiaKawamoto2013,KawamotoBL16,BiondiKLT17}.
There is also a tool, \leakiEst{}~\cite{chothia2013tool}, which allows to  estimate min-entropy leakage from a set of trial runs~\cite{ChothiaKN14:esorics,ChothiaKawamoto2014}.
However, if the cardinality of secrets and observables is large, such estimation becomes computationally heavy,  due to the huge amount of trial runs that need to be performed. 

In this paper we determine bounds on $g$-leakages in compositional terms,
focusing on the multiplicative version of $g$-leakage. 
More precisely, we consider the parallel composition of channels, defined on the cross-products of the inputs and of the outputs.
Then, we  derive lower and upper bounds on the $g$-leakage of the whole channel in terms of the $g$-leakages of the components. 
Since the size of the whole channel is the product of the sizes of the components, there is an evident  benefit in terms of computational cost. 
Table~\ref{table:channel_types} illustrates the situation for the various kinds of channel matrices (small, large white-box, large black-box): the first three rows 
characterize the situation, and the last three express the  feasibility of computing the leakage of the components, the matrix of the whole system, and the leakage of the whole system, respectively. This computation is meant to be exact in the first two columns, and statistical  in the last one.
The number of components is assumed to be huge. 
Note that the size of the whole channel matrix increases exponentially with the number of the components. 

\begin{table}[t]
\newcommand{\bhline}[1]{\noalign{\hrule height #1}}
\begin{center}
  \renewcommand{\arraystretch}{1.2}
 \begin{tabular}{|l|c|c|c|}
 \hline
 Kinds of systems			& small & large white-box & large black-box \\ 
 Input distribution $\pi$		& known	 	  & known		& known \\ 
 Component channels $C_i$	& known		  & known		&  {\small approx.\,statistically}  \\ \hline
 Leakage of $C_i$ with $\pi_i$ & computable	  & computable	& {\small approx.\,statistically} \\ 
 Composed channel $C$		& computable	  &unfeasible	&unfeasible \\ 
 Leakage of $C$ with $\pi$	& computable	  &{unfeasible} &unfeasible \\ \hline
 \end{tabular}
 \caption{Computation of information leakage measures in various scenarios.}
 \label{table:channel_types}
\end{center}
\end{table}

We  evaluate our compositionality results on randomly generated channels and 
on Crowds, a protocol for anonymous communication, run on top of a
mobile ad-hoc network (MANET). In such a network users are mobile, can
communicate only with nearby nodes, and the network topology changes
frequently. As a result, Crowds routes can become invalid forcing the user to
re-execute the protocol to establish a new route. These protocol repetitions,
modeled by the composition of the corresponding channels, lead to more
information being leaked. Although the composed channel quickly becomes too big
to compute the leakage directly, our compositionality results allow to obtain 
bounds on it. 

The results of this paper are about the multiplicative version of $g$-leakage, which is the original one considered in~\cite{Alvim:12:CSF}.
Some of the auxiliary results are about the posterior $g$-vulnerability.
There is also an additive version of $g$-leakage, investigated in~\cite{AlvimCMMPS14csf}.
At a preliminary analysis, the adaption of our results to the additive one is not trivial.
We leave this for future work.

%
%
The rest of the paper is organized as follows:
Section~\ref{sec:preliminaries} introduces basic notions of   information theory, defines compositions of channels, and presents information leakage measures.
Section~\ref{sec:compo-g-leak} presents lower\slash upper bounds for $g$-leakages in compositional terms.
Section~\ref{sec:compo-MEL} instantiates these results to  min-entropy leakages.
Section~\ref{sec:input-approx} introduces a transformation technique which improves the precision of our method.
Section~\ref{sec:channel-approx} presents a way of computing leakage bounds using the cascade composition refinement.
Section~\ref{sec:compo-MI} compares our results on $g$-leakage with the compositionality results on Shannon mutual information.
Section~\ref{sec:experiment} evaluates our results by experiments.
Section~\ref{sec:related} discusses related work.

A preliminary version of this paper, without proofs, appeared in~\cite{KawamotoCP14qest}.
This paper improves the bounds on $g$-leakage.
The new, tighter bounds are formally stated 
in Proposition~\ref{lem:disjoint-dep-g-prior-upper}, Lemma~\ref{lem:disjoint-dep-g-vul}, Theorem~\ref{thm:disjoint-dep-g-leak} and Proposition~\ref{thm:disjoint-dep-g-leak-k-channels}.
Furthermore, it extends some results to the case of $n$-ary composition
(Propositions~\ref{thm:disjoint-dep-g-leak-k-channels},~\ref{thm:g-leak-n-channels},~\ref{cor:comp-approx-MEL-joint-input-n-channels} and~\ref{cor:comp-approx-MEL-shared-input-n-channels})
and presents some new properties (Propositions~\ref{prop:Mgmin_positive},~\ref{lem:disjoint-dep-g-prior},~\ref{lem:disjoint-dep-g-prior-upper}).
Finally, Section~\ref{sec:channel-approx},~\ref{sec:compo-MI} and~\ref{sec:related} are new.

%
\section{Preliminaries}
\label{sec:preliminaries}
In this section we recall the notion of information-theoretic channels, define channel compositions/decomposition, and recall some information leakage measures.

\subsection{Channels}
\label{subsec:channels}
A \emph{discrete channel} is a triple $(\X, \Y, C)$ consisting of a finite set $\X$ of secret input values, a finite set $\Y$ of observable output values, 
and an $|\X| \times |\Y|$ matrix $C$, called  \emph{channel matrix}, where each element $C[x, y]$ represents the conditional probability $p(y | x)$ of obtaining the output $y \in \Y$ given the input $x \in \X$.
The input values have a probability distribution, called   \emph{input distribution} or   \emph{prior}. 
Given a prior $\pi$ on $\X$, the joint distribution for $X$ and $Y$ is defined by $p(x, y) = \pi[x] C[x, y]$. The output distribution is given by $p(y) = \displaystyle\sum_{x \in \X} \pi[x] C[x, y]$.

\subsection{Composition of Channels}
\label{subsec:compose}

%
%
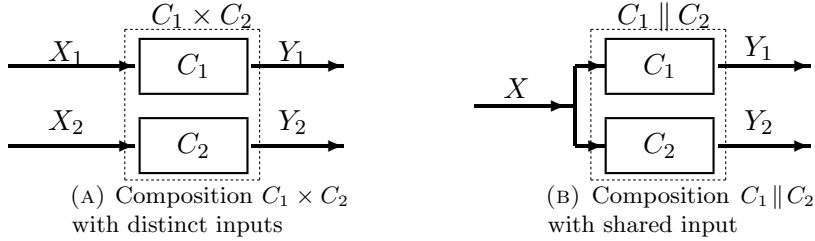
\begin{figure*}[t]\label{fig:compositions}
\hspace{-1.0ex}%
\subfloat[][Composition $C_1 \compd C_2$ with distinct inputs]{
\begin{picture}(169, 43)
 \put( 62, 46){$C_1 \compd C_2$}
 \thicklines \thicklines
 \put( 58, 20){\framebox(40,20){$C_1$}}
 \put( 58,-10){\framebox(40,20){$C_2$}}

 \linethickness{1.4pt}
 \put(   8,  30){\vector(  1,  0){48}}
 \put(   8,  00){\vector(  1,  0){48}}
 \put( 100,  30){\vector(  1,  0){35}}
 \put( 100,  00){\vector(  1,  0){35}}
 \thinlines \thinlines
 \put(  23,  32){$X_1$}
 \put(  23,    6){$X_2$}
 \put( 110,  32){$Y_1$}
 \put( 110,    6){$Y_2$}
 \put(52,-13){\dashbox{1.0}(51,57){}}
\end{picture}
\label{fig:composition-separated}
}
~~~
\subfloat[][Composition $C_1\!\compn\!C_2$ with shared input]{
\begin{picture}(169, 43)
 \put( 62,46){$C_1 \compn C_2$}
 \thicklines \thicklines
 \put( 58, 20){\framebox(40,20){$C_1$}}
 \put( 58,-10){\framebox(40,20){$C_2$}}

 \linethickness{1.4pt}
 \put(   8,  15){\line(  1,  0){38}}
 \put(   8,  15){\vector(  1,  0){35}}
 \put(  46,    0){\line(  0,  1){30}}
 \put(  46,  30){\vector(  1,  0){11}}
 \put(  46,    0){\vector(  1,  0){11}}
 \put( 100,  30){\vector(  1,  0){35}}
 \put( 100,    0){\vector(  1,  0){35}}
 \thinlines \thinlines
 \put(  19,  18){$X$}
 \put( 110,  34){$Y_1$}
 \put( 110,    6){$Y_2$}
 \put(52,-13){\dashbox{1.0}(51,57){}}
\end{picture}
\label{fig:composition-shared}
}
\caption{The two kinds of parallel compositions  on channels, $\compd$ and $\compn$.}
\label{fig:two-composition}
\end{figure*}

We now introduce the two kinds of compositions which will be considered in the paper. 
Figure~\ref{fig:two-composition} illustrates their definitions:
\begin{itemize}
\item
In the \emph{parallel composition $C_1\compd C_2$ with distinct inputs}, two channels $C_1$ and $C_2$ separately receive inputs drawn from two distributions $X_1$ and $X_2$, respectively (Figure~\ref{fig:composition-separated}).
\item
In the \emph{parallel composition $C_1\compn C_2$ with shared inputs}, $C_1$ and $C_2$ receive an identical input value drawn from a distribution $X$ (Figure~\ref{fig:composition-shared}).
\end{itemize}

In these compositions we assume that the channels are \emph{independent}, in the sense that, given the respective inputs,  the outcome of one channel does not influence the outcome of the other.

These two kinds of compositions are used to represent different situations.
For example, in the Crowds protocol (explained in Section~\ref{subsec:crowds}), an execution of the protocol is modeled as a channel $C$ from secret senders $\X$ to observable outputs $\Y$.
Then repeated executions of the protocol with (possibly) different senders are described by the parallel composition $C \compd C$, while repeated executions with the same sender are described by the parallel composition with shared input $C \compn C$.

We start with defining parallel composition with separate inputs $\compd$ (\emph{parallel composition} for short).
Note that the term ``parallel'' here does not carry a temporal meaning:
the actual execution of the corresponding systems could take place simultaneously or in any order.
\begin{defi}[Parallel composition (with distinct inputs)]\label{def:disjoint-compo} \rm
Given two discrete channels $(\X_1, \Y_1, C_1)$ and $(\X_2, \Y_2, C_2)$, their \emph{parallel composition (with distinct inputs)} is the discrete channel $(\X_1 \times \X_2, \Y_1 \times \Y_2, C_1\compd C_2)$ where $C_1\compd C_2$ is the $(|\X_1|\cdot |\X_2|) \times (|\Y_1|\cdot |\Y_2|)$ matrix such that for each $x_1 \in \X_1$, $x_2 \in \X_2$, $y_1 \in \Y_1$ and $y_2 \in \Y_2$,
\[
(C_1\compd C_2)[(x_1, x_2), (y_1, y_2)] = C_1[x_1, y_1] \cdot C_2[x_2, y_2].
\]
\end{defi}

The condition $(C_1\compd C_2)[(x_1, x_2), (y_1, y_2)] = C_1[x_1, y_1] \cdot C_2[x_2, y_2]$
is what we mean  by ``the channels are independent''. 
Note that, although the  output distributions $Y_1$ and $Y_2$ may be correlated, they are \emph{conditionally independent}, in the sense that
$p(y_1, y_2 | x_1,x_2) = p(y_1 | x_1) p(y_2 | x_2)$.

Next, we define  the parallel composition with shared input  $\compn$.
\begin{defi}[Parallel composition with shared input] \label{def:parallel-compo} \rm 
Given two discrete channels $(\X, \Y_1, C_1)$ and $(\X, \Y_2, C_2)$, their \emph{parallel composition with shared input} is the discrete channel $(\X, \Y_1 \times \Y_2, C_1\compn C_2)$ where $C_1\compn C_2$ is the $|\X| \times (|\Y_1|\cdot |\Y_2|)$ matrix such that for each $x \in \X$, $y_1 \in \Y_1$ and $y_2 \in \Y_2$,
\[
(C_1\compn C_2)[x, (y_1, y_2)] = C_1[x, y_1] \cdot C_2[x, y_2].
\]
\end{defi}
Note that  $\compn$  is a special case of $\compd$. In fact,
$(C_1\compn C_2)[x, (y_1, y_2)] = C_1[x, y_1] \cdot C_2[x, y_2] = (C_1\compd C_2)[(x,x), (y_1, y_2)]$.

\subsection{Decomposition of Channels}
\label{subsec:decompose}
Sometimes it can be useful to decompose a large system for analysis purposes: the components may be easier to analyze than the whole system. 
Here, given a  channel that outputs pairs in $\Y_1\times\Y_2$,   we consider the sub-channels that output on $\Y_1$ and $\Y_2$
separately. Depending on whether we wish to perform an analogous separation also on the inputs or not, we obtain two kinds of decomposition. 
\begin{defi}[Decomposition] \rm\label{def:de-compo-o}
The \emph{decomposition} of a discrete channel $(\X, \Y_1 \times \Y_2, C)$ with respect to $\Y_1$ and $\Y_2$ 
is a pair of channels  $(\X, \Y_1, C|_{\Y_1})$ and $(\X, \Y_2, C|_{\Y_2})$ where
\begin{align*}
C|_{\Y_1}[x, y_1] &= \sum_{y_2 \in \Y_2} C[x, (y_1, y_2)] \\
C|_{\Y_2}[x, y_2] &= \sum_{y_1 \in \Y_1} C[x, (y_1, y_2)].
\end{align*}
\end{defi}
For a channel $(\X, \Y_1 \times \Y_2, C_1 \compn C_2)$ with shared input, we have $(C_1 \compn C_2)|_{\Y_i} = C_i$ for each $i\in\set{1,2}$.

\subsection{Quantitative Information Leakage Measures}
\label{subsec:leak-measures}
The \emph{information leakage} of a channel is measured as the difference between the \emph{prior uncertainty} about the secret value of the channel's input and the \emph{posterior uncertainty} of the input  after observing the channel's output.
The uncertainty is defined in terms of an attacker's operational scenario. 
In this paper we will focus on  \emph{min-entropy leakage}, in which such measure, min-entropy,  represents the difficulty for an attacker to guess the secret inputs in a single attempt.
\begin{defi}[Vulnerability] \label{def:vulnerability} \rm
Given a prior $\pi$ on $\X$ and a channel $(\X, \Y, C)$, the \emph{prior vulnerability} and the 
 the \emph{posterior vulnerability}  are defined respectively as 
\begin{align*}
V(\pi) &= \max_{x \in \X} \pi[x] \\
V(\pi, C) &= \sum_{y \in \Y} \max_{x \in \X} \pi[x] C[x, y]
\texttt{.}
\end{align*}
\end{defi}
\begin{defi}[Min-entropy leakage] \label{def:min-entropy-leak} \rm
Given a prior $\pi$ on $\X$ and a channel $(\X, \Y, C)$, the \emph{min-entropy} $\Hinf(\pi)$ and \emph{conditional min-entropy} $\Hinf(\pi, C)$ are defined by:
\begin{align*}
\Hinf(\pi) &= -\log V(\pi) \\
\Hinf(\pi, C) &= -\log V(\pi, C)
\texttt{.}
\end{align*}
The \emph{min-entropy leakage} $\Iinf(\pi, C)$ and \emph{min-capacity} $\Cinf(C)$ are defined by:
\begin{align*}
\Iinf(\pi, C) &= \Hinf(\pi) - \Hinf(\pi, C) \\
\Cinf(C) &= \displaystyle \sup_{\pi'} \Iinf(\pi', C).
\end{align*}
\end{defi}

Min-entropy leakage has been generalized by \emph{$g$-leakage}~\cite{Alvim:12:CSF}, which allows a wide variety of operational scenarios of attacks.
These are modeled using a set $\W$ of possible \emph{guesses}, and a \emph{gain function}
$g: \W\times\X  \rightarrow [0, 1] $ such that $g(w, x)$ represents the gain of the attacker when the secret value is $x$ and he makes a guess $w$ on $x$.

Then \emph{$g$-vulnerability} is defined as the maximum expected gain of the attacker:
\begin{defi}[$g$-vulnerability] \label{def:g-vulnerability} \rm
Given a prior $\pi$ on $\X$ and a channel $(\X, \Y, C)$, the \emph{prior $g$-vulnerability} and the \emph{posterior $g$-vulnerability} are defined respectively by
\begin{align*}
V_{g}(\pi) &= \max_{w \in \W} \sum_{x \in \X} \pi[x] g(w, x) \\
V_{g}(\pi, C) &= \sum_{y \in \Y} \max_{w \in \W} \sum_{x \in \X} \pi[x] C[x, y] g(w, x)
\texttt{.}
\end{align*}
\end{defi}

We now extend Definition~\ref{def:min-entropy-leak} to the $g$-setting:
\begin{defi}[$g$-leakage, $g$-capacity] \label{def:g-leak} \rm
Given a prior $\pi$ on $\X$ and a channel $(\X, \Y, C)$, the \emph{$g$-entropy} $\Hg(\pi)$, \emph{conditional $g$-entropy}  $\Hg(\pi, C)$ and \emph{$g$-leakage} $\Ig(\pi, C)$ are defined by:
%
\begin{align*}
\Hg{\small (\pi)} &= -\log V_{\!g}{\small (\pi)}, \\
\Hg{\small (\pi, C)} &= -\log V_{\!g}{\small (\pi, C)}, \\
\Ig{\small (\pi, C)} &= \Hg{\small (\pi)}\,\!-\!\,\Hg{\small (\pi, C)}
\texttt{.}
\end{align*}
The \emph{$g$-capacity} $\Cg(C)$ is defined as the maximum $g$-leakage over all priors.
\begin{align*}
\Cg{\small (C)} &= \displaystyle\sup_{\pi'}\Ig{\small (\pi', C)}
\texttt{.}
\end{align*}
\end{defi}

Here are a few examples of gain functions introduced in~\cite{Alvim:12:CSF}.
\begin{exa}[Identity gain function]\label{eg:id-gain}
The \emph{identity gain function} represents the situation in which an attacker's guess $w$ is worthy only when it is the correct secret value $x$.
Formally, it is defined as the following function $g_{id}: \X\times\X\rightarrow[0,1]$:
\[
g_{id}(w, x) = \left\{
\begin{array}{l}
1 ~~~\mbox{ if $x = w$}\\
0 ~~~\mbox{ otherwise.}\\
\end{array}
\right.
\]
\noindent
The min-entropy notions are particular cases of the  $g$-entropy ones, obtained by instantiating $g$ to $g_{id}$; i.e., $\Hinf = H_{g_{id}}$,  $\Iinf  = I_{g_{id}} $ and $\Cinf = C_{g_{id}}$.
\end{exa}

\begin{exa}[Binary gain functions]\label{eg:binary-gain}
The gain functions returning boolean values are called \emph{binary}.
They are parameterized with a collection $\W \subseteq 2^\X$ of guess sets such that the attacker gains when the correct secret belongs to his guessed set $w\in\W$.
Intuitively, they represent the gain of the adversary by guessing a secret partially.
Formally, it is defined as the following function $g_{\W}: \W\times\X\rightarrow[0,1]$:
\[
g_{\W}(w, x) = \left\{
\begin{array}{l}
1 ~~~\mbox{ if $x \in w$}\\
0 ~~~\mbox{ otherwise.}\\
\end{array}
\right.
\]
The identity gain function $g_{id}$ in Example~\ref{eg:id-gain} is a particular case of a binary gain function.
\end{exa}

\begin{exa}[$k$-tries gain function]\label{eg:k-ties-gain}
For any positive integer $k$, the \emph{$k$-tries gain function} models a scenario in which an attacker is allowed to make guesses on the secret $k$ times.
Formally, the set of guesses is defined by:
\[
\W_k = \{ w \in 2^\X \colon \#w \le k \}
\texttt{.}
\]
Then the $k$-tries gain function is defined as the binary gain function $g_{\W_k}$.
\end{exa}

%
%
\section{Compositionality Results on $g$-Leakage}
\label{sec:compo-g-leak}
In this section we present compositionality results for $g$-leakage.
To obtain them we introduce the notions of joint gain functions for composed channels (Section~\ref{subsec:joint-gain}) and of jointly supported distributions (Section~\ref{subsec:jointly-supported}).
Then we show compositionality theorems for $g$-leakage that provide a certain lower and an upper bound on the $g$-leakage of a composed system (Sections~\ref{subsec:results-g-leakage} and~\ref{subsec:results-g-leakage-shared-input}).
Finally, we extend these results from binary to $n$-ary channels (Section~\ref{subsec:extension-to-n}).
The summary of these results are shown in Table~\ref{table:compositionality:g-leak}.
\begin{table}[t]
\newcommand{\bhline}[1]{\noalign{\hrule height #1}}
\begin{center}
\begin{small}
  \def\arraystretch{1.2}
 \begin{tabular}{|l|l|c|c|}
 \hline
 Composition & Measure					& Lower bound		& Upper bound \\[0.4ex] \bhline{0.3mm}
 & prior $g$-entropy
 	& Proposition~\ref{lem:disjoint-dep-g-prior}& Proposition~\ref{lem:disjoint-dep-g-prior-upper}
 \\[0.4ex] \cline{2-4}
 ~Composition	& posterior $g$-entropy
 	& Lemma~\ref{lem:disjoint-dep-g-vul} (1)	& Lemma~\ref{lem:disjoint-dep-g-vul} (2) \\[0.4ex]
 ~with distinct &\hspace{6ex}\scriptsize{with independent priors}& \multicolumn{2}{c|}{Corollary~\ref{cor:disjoint-g-vul}} \\[0.4ex] \cline{2-4}
 ~input $\compd$& $g$-leakage
 	& Theorem~\ref{thm:disjoint-dep-g-leak} (1) & Theorem~\ref{thm:disjoint-dep-g-leak} (2)  \\[0.4ex]
 &\hspace{6ex}\scriptsize{with independent priors}
 	& \multicolumn{2}{c|}{Corollary~\ref{cor:disjoint-g-leak}} \\[0.4ex]
 & \hspace{11ex}\scriptsize{$n$-ary composition}
 	& Proposition~\ref{thm:disjoint-dep-g-leak-k-channels}& Proposition~\ref{thm:disjoint-dep-g-leak-k-channels} \\[0.4ex] \hline
 ~Composition  & posterior $g$-entropy
 	& Theorem~\ref{thm:g-vul}			& {---} \\[0.4ex]
 ~with shared & $g$-leakage
 	& {---}			& Theorem~\ref{thm:g-leak}  \\[0.4ex]
 ~input $\compn$& \hspace{11ex}\scriptsize{$n$-ary composition}
	 & {---}	& Proposition~\ref{thm:g-leak-n-channels} \\[0.4ex] \hline
 \end{tabular}
 \caption{Compositionality results for $g$-leakage.}
 \label{table:compositionality:g-leak}
\end{small}
\end{center}
\end{table}
\subsection{Joint Gain Functions for Composed Channels}
\label{subsec:joint-gain}

In this section we introduce the notion of \emph{joint gain functions}, which will be used to formalize the $g$-leakages of composed channels.

A \emph{joint gain function} $g$ is defined as a function from $(\W_1 \times \W_2) \times (\X_1 \times \X_2)$ to $[0 ,1]$.
When a joint secret input is $(x_1, x_2) \in \X_1 \times \X_2$ and the attacker's joint guess is $(w_1, w_2) \in \W_1 \times \W_2$,\, 
the attacker's joint gain from the guesses is represented by $g((w_1, w_2), (x_1, x_2))$.

For the sake of generality, we do not assume any relation between  $g$ and  the two gain functions $g_1$ and $g_2$, except for the following:
a joint guess is worthless iff at least one of the single guesses is worthless. 
Formally:  $g((w_1, w_2), (x_1, x_2)) = 0$ iff $g_1(w_1, x_1) g_2(w_2, x_2) = 0$.~%
\footnote{This property holds, for example, when $g, g_1, g_2$ are the identity gain functions.}

We say that $g_1$ and $g_2$ are \emph{independent} if $g((w_1, w_2), (x_1, x_2)) = g_1(w_1, x_1) \allowbreak g_2(w_2, x_2)$ for all $x_1,x_2,w_1$ and $w_2$.

\subsection{Jointly Supported Distributions}
\label{subsec:jointly-supported}
In this section we define some notions on probability distributions to obtain compositionality results on $g$-leakage.

Given a joint distribution $\pi$ on $\X_1 \times \X_2$,  the   \emph{marginal distribution} $\pi_1$ on $\X_1$ is defined as
$\pi_1[x_1] = \sum_{x_2 \in \X_2} \pi[x_1, x_2]$ for all $x_1 \in \X_1$. 
The \emph{marginal distribution} $\pi_2$ on $\X_2$  is defined analogously.
Note that $\pi_1[x_1] \cdot \pi_2[x_2] = 0$ implies $\pi[x_1, x_2] = 0$ . The converse does not hold in general, but occasionally 
we will assume it:
\begin{defi} \label{def:jointly-supported} \rm
A joint distribution $\pi$ on $\X_1 \times \X_2$ is \emph{jointly supported} if, for all $x_1 \in \X_1$ and $x_2 \in \X_2$,\, $\pi_1[x_1] \cdot \pi_2[x_2] \neq 0$ implies $\pi[x_1, x_2] \neq 0$.
\end{defi}

Essentially, this condition rules out all the distributions in which there exist two events that happen with a non-zero probability, but that never happen together, i.e., events that are incompatible with each other.

For example, $\pi$ is jointly supported when $\pi_1$ and $\pi_2$  are independent; i.e., $\pi[x_1, x_2] = \pi_1[x_1] \cdot \pi_2[x_2]$ for all $x_1 \in \X_1$ and $x_2 \in \X_2$.
Hereafter we denote such $\pi$ by $\pi_1 \times \pi_2$.

\subsection{The $g$-Leakage of Parallel Composition with Distinct Input}
\label{subsec:results-g-leakage}
In this section we present a lower and an upper bound for the $g$-leakage of  $C_1 \compd C_2$ in terms of the $g$-leakages of $C_1$ and $C_2$.
We first introduce some notations to represent the bounds (Section~\ref{subsubsec:measures}). 
Then we present compositionality results on the prior $g$-entropy $\Hg(\pi)$ (Section~\ref{subsubsec:g-prior-of-paralell}), the posterior $g$-entropy $\Hg(\pi, C_1 \compd C_2)$ (Section~\ref{subsubsec:g-post-of-paralell}) and the $g$-leakage $\Ig(\pi, C_1 \compd C_2)$ (Section~\ref{subsubsec:g-leak-of-paralell}).

\subsubsection{Measures $\Mminpi$ and $\Mmaxpi$}
\label{subsubsec:measures}
We introduce two measures $\Mminpi$ and $\Mmaxpi$ that capture certain relationships between the joint prior/gain function and their marginals, which will be used to provide the lower and upper bounds on the $g$-leakage of composed channels.
These measures are defined using the following set:
\begin{defi} \label{def:full-support-set} \rm
Let $\pi$ be a prior on $\X_1 \times \X_2$, and $g:(\W_1 \times \W_2) \times (\X_1 \times \X_2)\rightarrow [0 ,1]$ be a joint gain function. For  $w_1 \in \W_1$ and $w_2 \in \W_2$, their   \emph{support with respect to} $g$ is defined as:
\begin{align*}
\S_{w_1, w_2} =
\left\{ (x_1, x_2) \in \X_1 \times \X_2 \,|\,
 \pi[x_1, x_2] \cdot g((w_1, w_2), (x_1, x_2)) \neq 0
\right\}.
\end{align*}
\end{defi}
The two measures $\Mminpi$ and $\Mmaxpi$ are defined as follows.
\begin{defi} \label{def:input-ratio} \rm
Let $g$ be a joint gain function from $(\W_1 \times \W_2) \times (\X_1 \times \X_2)$ to $[0 ,1]$. Let $g_1, g_2$ be two gain functions from $\W_1  \times \X_1$ to $[0 ,1]$ and from $\W_2 \times \X_2$ to $[0 ,1]$ respectively.
Given a prior $\pi$ on $\X_1 \times \X_2$,
we define $\Mminpi$ and $\Mmaxpi$ by:
\begin{align*}
\Mminpi &= 
\hspace{-0.0ex}
\min_{w_1 \in \W_1, w_2 \in \W_2}
\hspace{0.5ex}
\min_{(x_1, x_2) \in \S_{w_1, w_2}}
\hspace{-0.5ex}
\frac{ \pi_1[x_1] g_1(w_1, x_1) \cdot \pi_2[x_2]  g_2(w_2, x_2) }{ \pi[x_1, x_2] \cdot g((w_1, w_2), (x_1, x_2)) }
\\[2.5ex]
\Mmaxpi &= 
\hspace{-0.0ex}
\max_{w_1 \in \W_1, w_2 \in \W_2}
\hspace{0.5ex}
\max_{(x_1, x_2) \in \S_{w_1, w_2}}
\hspace{-0.5ex}
\frac{ \pi_1[x_1] g_1(w_1, x_1) \cdot \pi_2[x_2] g_2(w_2, x_2) }{ \pi[x_1, x_2] \cdot  g((w_1, w_2), (x_1, x_2)) }.
\end{align*}
\end{defi}
\vspace{1ex}
Intuitively, $\Mminpi$  and $\Mmaxpi$ represent how the joint distribution $\pi$ and joint gain function $g$ are far from the cases where $\pi_1$ and $g_1$ are respectively independent of $\pi_2$ and of $g_2$.
Computing these values do not involve the channel matrix, and has a time complexity of ${\mathcal O}(\#\W_1\times\#\W_2\times\#\X_1\times\#\X_2)$, which is smaller than the complexity of computing the $g$-leakage of the composed channel: ${\mathcal O}(\#\W_1\times\#\W_2\times\#\X_1\times\#\X_2\times\#\Y_1\times\#\Y_2)$.
Note that the definition of $\Mmaxpi$ is different from the preliminary version~\cite{KawamotoCP14qest} of this work in that we take the maximum over $\S_{w_1, w_2}$ instead of the summation over them. 
Thanks to this improvement we obtain better compositionality results (Proposition~\ref{lem:disjoint-dep-g-prior-upper}, Lemma~\ref{lem:disjoint-dep-g-vul}, Theorem~\ref{thm:disjoint-dep-g-leak} and Proposition~\ref{thm:disjoint-dep-g-leak-k-channels}).

These two measures $\Mminpi$ and $\Mmaxpi$ satisfy the following properties.
When $\pi_1$ and $\pi_2$ are independent and $g_1$ and $g_2$ are independent, $\Mminpi = \Mmaxpi =~1$.
In addition, for any prior $\pi$, $\Mminpi$ is strictly positive.
\begin{prop} \label{prop:Mgmin_positive}
For any prior $\pi$ and any joint gain function $g$,\, $\Mminpi > 0$.
\end{prop}
%
\proof
Let $\pi$ be any prior on $\X_1\times\X_2$ and $g$ be any gain function from $(\W_1 \times \W_2) \times (\X_1 \times \X_2)$ to $[0 ,1]$.
Since probabilities and gains are non-negative, it suffices to show that $\Mminpi \neq 0$.
Let $w_1 \in \W_1$, $w_2 \in \W_2$ and $(x_1, x_2) \in \S_{w_1, w_2}$.
We have $\pi[x_1, x_2] \neq 0$ and $g((w_1, w_2), (x_1, x_2)) \neq 0$.
Then, by the definition of joint gain functions, 
$g_1(w_1, x_1) g_2(w_2, x_2) \neq 0$.
In addition, we obtain $\pi_1[x_1] \cdot \pi_2[x_2] \neq 0$ from $\pi[x_1, x_2] \neq 0$.
Hence $\pi_1[x_1] g_1(w_1, x_1) \cdot \pi_2[x_2]  g_2(w_2, x_2) \neq 0$
Therefore $\Mminpi \neq 0$.
\qed
\medskip

\subsubsection{The $g$-Entropy of Joint Prior}
\label{subsubsec:g-prior-of-paralell}
In this section we present a lower and an upper bound on the (prior) $g$-entropy $\Hg(\pi)$.

The $g$-entropy $\Hg(\pi)$ of a joint prior $\pi$ is bounded from below by the summation of $\log \Mminpi$ and the $g$-entropies of its two marginals $\pi_1, \pi_2$:
\begin{prop} \label{lem:disjoint-dep-g-prior}
For any prior $\pi$ on $\X_1 \times \X_2$,
\[
\Hg(\pi) \ge
H_{g_1}(\pi_1) + H_{g_2}(\pi_2) + \log \Mminpi.
\]
\end{prop}
\proof
See Appendix~\ref{subsec:proofs}.
\vspace{2ex}

To obtain an upper bound for the $g$-entropy $\Hg(\pi)$ of a joint prior $\pi$, we assume that $\pi$ is jointly supported.
Then the upper bound is the summation of $\log \Mmaxpi$ and the $g$-entropies of its two marginals $\pi_1, \pi_2$:
\begin{prop} \label{lem:disjoint-dep-g-prior-upper}
For any jointly supported prior $\pi$ on $\X_1 \times \X_2$,
\[
\Hg(\pi) \le
H_{g_1}(\pi_1) + H_{g_2}(\pi_2) + \log \Mmaxpi.
\]
\end{prop}
\proof
See Appendix~\ref{subsec:proofs}.

\subsubsection{The Posterior $g$-Entropy of Parallel Composition}
\label{subsubsec:g-post-of-paralell}
Next we show a lower and an upper bound on the posterior $g$-entropy $\Hg(\pi, C_1 \compd C_2)$.
\begin{lem} \label{lem:disjoint-dep-g-vul}
For any prior $\pi$ on $\X_1 \times \X_2$ with marginals $\pi_1$ and $\pi_2$, and two channels $(\X_1, \Y_1, C_1)$, $(\X_2, \Y_2, C_2)$,
\begin{enumerate}
\item 
$\Hg(\pi, C_1 \compd C_2) \ge H_{g_1}(\pi_1, C_1) + H_{g_2}(\pi_2, C_2) + \log \Mminpi$
\item 
$\Hg(\pi, C_1 \compd C_2) \le H_{g_1}\!(\pi_1, C_1) + H_{g_2}\!(\pi_2, C_2) + \log\Mmaxpi$
if $\pi$ is jointly supported.
\end{enumerate}
\end{lem}
\Proof.
See Appendix~\ref{subsec:proofs}.
\vspace{2ex}

The equalities in Lemma~\ref{lem:disjoint-dep-g-vul} hold if the priors and the gain functions are independent:
\begin{cor} \label{cor:disjoint-g-vul}
If $g((w_1, w_2), (x_1, x_2)) = g_1(w_1, x_1) g_2(w_2, x_2)$ for all $x_1,x_2,w_1$ and $w_2$, then, for any  $\pi_1$ and $\pi_2$,
we have
\[
\Hg(\pi_1 \times \pi_2, C_1 \compd C_2) = H_{g_1}(\pi_1, C_1) + H_{g_2}(\pi_2, C_2)
\texttt{.}
\]
\end{cor}
%
\proof
By the independence of the priors and the gain functions, we have $\Mminpi = \Mmaxpi =~1$.
Hence the claim follows.
\qed

%
\subsubsection{The $g$-Leakage of Parallel Composition}
\label{subsubsec:g-leak-of-paralell}
Now we show a lower and an upper bound for the $g$-leakage of the parallel composition. To obtain this we assume that the prior is jointly supported.

\begin{thm}\label{thm:disjoint-dep-g-leak}
Let $\pi$ be a  jointly supported prior on $\X_1 \times \X_2$ with marginals $\pi_1$ and $\pi_2$. Let $(\X_1, \Y_1, \allowbreak C_1)$, $(\X_2, \Y_2, C_2)$ be two channels. Then:
\begin{enumerate}
\item 
$\Ig(\pi, C_1 \compd C_2) \ge I_{g_1}(\pi_1, C_1) + I_{g_2}(\pi_2, C_2)  - \log\!\frac{ \Mmaxpi }{ \Mminpi }$,
\item 
$\Ig(\pi, C_1 \compd C_2) \le I_{g_1}\!(\pi_1,\!C_1) + I_{g_2}\!(\pi_2,\!C_2)  + \log\!\frac{ \Mmaxpi }{ \Mminpi }$.
\end{enumerate}
\end{thm}
%
\proof
\begin{align*}
&\Ig(\pi, C_1 \compd C_2) \\[0.5ex]
=\,&
\Hg(\pi) - \Hg(\pi, C_1 \compd C_2)
\\ \le\,&
H_{g}(\pi) - H_{g_1}(\pi_1, C_1) - H_{g_2}(\pi_2, C_2) - \log \Mminpi
&\hspace{-20ex}\mbox{(By Lemma~\ref{lem:disjoint-dep-g-vul})}
\\[0.5ex] \le\,&
\left( H_{g_1}(\pi_1) - H_{g_1}(\pi_1, C_1) \right) + \left( H_{g_2}(\pi_2) - H_{g_2}(\pi_2, C_2) \right)
+ \log \Mmaxpi - \log \Mminpi
\\ &
&\hspace{-20ex}\mbox{(By Proposition~\ref{lem:disjoint-dep-g-prior-upper})}
\\[0.5ex] =\,&
I_{g_1}(\pi_1, C_1) + I_{g_2}(\pi_2, C_2)  + \log \frac{ \Mmaxpi }{ \Mminpi }.
\end{align*}
The other inequality is proven using Lemma~\ref{lem:disjoint-dep-g-vul} and Proposition~\ref{lem:disjoint-dep-g-prior} in an analogous way.
\qed
%
Again, the equality holds if the priors and gain functions are independent:
\begin{cor} \label{cor:disjoint-g-leak}
If two gain functions $g_1$ and $g_2$ are independent, 
then 
$
\Ig(\pi_1 \times \pi_2, C_1 \compd C_2) = I_{g_1}(\pi_1, C_1) \allowbreak + I_{g_2}(\pi_2, C_2)
\texttt{.}
$
\end{cor}
%
\proof
By the independence of the priors and the gain functions, we have $\Mminpi = \Mmaxpi =~1$.
Hence the claim follows.
\qed

%
\subsection{The $g$-Leakage of Parallel Composition with Shared Input}
\label{subsec:results-g-leakage-shared-input}
In this section we present compositionality results for $g$-leakage when two channels share the same  input value. 

The parallel composition with shared input corresponds to the parallel composition with two identical inputs values: 
$(C_1 \compn C_2)[x, (y_1, y_2)] = C_1[x, y_1] C_2[x, y_2] = (C_1 \compd C_2)[(x, x), (y_1, y_2)]$.
To give the same input value $x$ to both $C_1$ and $C_2$, 
the prior $\pid$ on $\X \times \X$ is defined from a prior $\pi$ on $\X$ by:
\[
\pid[x, x'] = \left\{
\begin{array}{lll}
\pi[x] &\quad&\mbox{if $x = x'$}\\
0 &&\mbox{otherwise}
\end{array}
\right.
\]
Then $\Hg(\pi, C_1 \compn C_2) = \Hg(\pid, C_1 \compd C_2)$.
In addition, 
$\pid_1[x] = \pid_2[x] = \pi[x]$.

As we see in the definition, the attacker's gain is determined solely from a secret input $x$ and his guess $w$ on $x$ (and independently of channels that receive $x$ as input).
Let $g$ be a gain function from $\W \times \X$ to $[0, 1]$.
Since $C_1$ and $C_2$ receive input from the same domain $\X$, we use the same gain function $g$ to calculate both the $g$-leakages of $C_1$ and $C_2$.
Since an identical input value $x$ is given to $C_1$ and $C_2$ in the composed channel $C_1 \compn C_2$ and the attacker makes a single guess $w$ on the secret $x$, we define the joint gain function $\gd \colon \W \times\W \times \X\times\X \rightarrow [0, 1]$ from $g$ by:
\[
\gd((w, w'), (x, x')) = \left\{
\begin{array}{l}
g(w, x) ~~~~~~\hfil\mbox{(if $w = w'$ and $x = x'$)}\\
0 ~~~~~~\hspace{5.0ex}\mbox{(otherwise)}
\end{array}
\right.
\]
If $\pid[x, x'] \cdot \gd((w, w'), (x, x')) \neq 0$, then $w = w'$ and $x = x'$.
Let $(\W \times \X)^+ = \set{ (w, x) \in \W \times \X \mid \pi[x] g(w, x) \neq 0 }$.
By $\pid_1[x] = \pid_2[x] = \pi[x]$,
$\Mmin{\pid} = \min_{(w, x) \in (\W \times \X)^+} \pi[x] g(w, x)$ and
$\Mmax{\pid} = \max_{w \in \W} \sum_{x \in \X} \pi[x] g(w, x)$.
Then $\Hg(\pi) = -\log \Mmax{\pid}$.

To describe compositionality results, we introduce the following notation.
\begin{defi} \label{def:Hgmin} \rm
For any prior $\pi$ on $\X$ and any gain function $g \colon \W \times \X \rightarrow [0, 1]$, we define $\HgMin(\pi)$ by:
\begin{align*}
\HgMin(\pi) = -\log \min \left\{ \pi[x] g(w, x) \colon
x \in \X, w \in \W,\, \pi[x] g(w, x) \neq 0 \right\}
\texttt{.}
\end{align*}
\end{defi}
Then, for any prior $\pi$,\, $\HgMin(\pi) = -\log \Mmin{\pid}$ and $\HgMin(\pi) \ge\displaystyle  \Hg(\pi)$.
Note that computing $\HgMin(\pi)$ does not involve the channel matrix, and has a time complexity of ${\mathcal O}(\#\W\times\#\X)$, which is smaller than the complexity of computing the $g$-leakage of the composed channel: ${\mathcal O}(\#\W\times\#\X\times\#\Y_1\times\#\Y_2)$.
\medskip
\medskip

Since $\pid$ is \emph{not} jointly supported, we can instantiate compositionality results only on a lower bound for the posterior $g$-entropy in Lemma~\ref{lem:disjoint-dep-g-vul} (1).
The lower bound for the posterior $g$-entropy $\Hg(\pi, C_1 \compn C_2)$ of a channel composed in parallel with shared inputs is described using $\HgMin(\pi)$ and the posterior $g$-entropies of its two components:
\begin{thm} \label{thm:g-vul}
For any prior $\pi$ on $\X$ and two channels $(\X, \Y_1, C_1)$ and $(\X, \Y_2, C_2)$,
\[
\Hg(\pi, C_1 \compn C_2) \ge
H_{g}(\pi, C_1) + H_{g}(\pi, C_2) - \HgMin(\pi).
\]
\end{thm}
%
\proof
We instantiate Lemma~\ref{lem:disjoint-dep-g-vul} by taking $g_1$ and $g_2$ as $g$, $\pi$ as $\pid$ and $\X_1 = \X_2 = \X$:
\[
\Hg(\pid, C_1 \compd C_2) \ge
H_{g}(\pid_1, C_1) + H_{g}(\pid_2, C_2) + \log \Mmin{\pid}.
\]
By $\pid_1[x] = \pid_2[x] = \pi[x]$ for all $x \in \X$,  $\log\Mmin{\pid} = -\HgMin(\pi)$ and $\Hg(\pi, C_1 \compn C_2) = \Hg(\pid, C_1 \compd C_2)$, the theorem follows.
\qed

An upper bound on the $g$-leakage $\Ig(\pi, C_1 \compn C_2)$ of a channel composed in parallel with shared inputs is described using the $g$-leakages of its two components in the following theorem.
On the other hand, since $\pid$ is not jointly supported, we cannot instantiate Theorem~\ref{thm:disjoint-dep-g-leak} hence cannot obtain a lower bound for the leakage.
\begin{thm} \label{thm:g-leak}
For any prior $\pi$ on $\X$ and two channels $(\X, \Y_1, C_1)$ and $(\X, \Y_2, C_2)$,
\pagebreak[2]
\[
\Ig(\pi, C_1 \compn C_2) \le
I_{g}(\pi, C_1) + I_{g}(\pi, C_2) + \HgMin(\pi) - \Hg(\pi).
\]
\end{thm}
\begin{proof}
\begin{align*}
&
\Ig(\pi, C_1 \compn C_2)
\\ =\, &
\Hg(\pi) - \Hg(\pi, C_1 \compn C_2)
\\ \le\, &
H_{g}(\pi) - H_{g}(\pi, C_1) - H_{g}(\pi, C_2) + \HgMin(\pi)
&
\hspace{-15ex}\mbox{(By Theorem~\ref{thm:g-vul})}
\\[0.5ex] \le\, &
\left( H_{g}(\pi) - H_{g}(\pi, C_1) \right) + \left( H_{g}(\pi) - H_{g}(\pi, C_2) \right)
- H_{g}(\pi) + \HgMin(\pi)
\\ =\, &
I_{g}(\pi, C_1) + I_{g}(\pi, C_2)  + \HgMin(\pi) - \Hg(\pi). \tag*{\qEd}
\end{align*}
\def\popQED{}
\end{proof}

We emphasize this result holds for any prior.
Note that in the right-hand side of the above inequality, $\HgMin(\pi) - \Hg(\pi)$ is necessary as the following illustrates.
\begin{exa} \label{eg:g-leak}
Let us consider the channel $(\X, \Y, C)$ where $\X = \Y = \set{0, 1}$ and $C$ is the $2 \times 2$ matrix defined by $C[0, 0] = C[1, 1] = 0.9$ and $C[0, 1] = C[1, 0] = 0.1$.
Let $g$ be the identity gain function $g_{id}$ and $\pi$ be the prior on $\X$ such that $\pi[0] = 0.1$ and $\pi[1] = 0.9$.
Then $\Hg(\pi) = \Hg(\pi, C) = -\log 0.9$, $\Hg(\pi, C \compn C) = -\log 0.972$.
Therefore $\Ig(\pi, C \compn C) = \log 1.08 > 0 = \Ig(\pi, C) + \Ig(\pi, C)$.

Note that the inequality of Theorem~\ref{thm:g-leak} does not give a useful upper bound when the prior $\pi$ is very far from the uniform distribution.
In this example, by $\HgMin(\pi) - \Hg(\pi) = \log 9$, the left-hand side is $\log 1.08 \approx 0.111$ while the right-hand side is $\log 9 \approx 3.170$.
\end{exa}

On the other hand, the result may not hold when the composition of channels is done in a dependent way, which is not a parallel composition.
The following is a counterexample:
\begin{exa} \label{prop:non-compositionality}
Let $\X = \Y_1 = \Y_2 = \{0, 1\}$, $\pi$ be the uniform distribution on $\X$ and $g$ be the identity gain function.
We consider the channel that, given an input $x \in \X$, outputs a bit $y_1$ uniformly drawn from $\Y_1$ and the exclusive OR $y_2$ of $x$ and $y_1$.
Then the $g$-leakage of the channel is $1$ while both of the $g$-leakages from $\X$ to $\Y_1$ and from $\X$ to $\Y_2$ are $0$ and $\HgMin(\pi) - \Hg(\pi) = 0$.
Hence the property expressed by Theorem~\ref{thm:g-leak} in general does not hold if we replace $\compn$ with some other kind of composition.
\end{exa}

\subsection{Extension to $n$-ary composition}
\label{subsec:extension-to-n}
In this section we extend some of the previous compositionality results to the case of $n$-ary parallel compositions.
Let $\I = \{1, 2, \ldots, n\}$ be the index set.
The $n$-ary compositions, with distinct and shared inputs, are denoted by:
$\prod_{i\in\I} C_i = C_1 \compd C_2 \compd \cdots \compd C_n$ and
$\parallel_{i\in\I} C_i = C_1 \compn C_2 \compn \cdots \compn C_n$.
The definitions of $\Mminpi$ and $\Mmaxpi$ are extended to $n$-ary compositions as follows.
\begin{defi} \label{def:input-ratio-n-channels} \rm
Let $\W = \W_1 \times \cdots \times \W_n$,
$\X = \X_1 \times \cdots \times \X_n$
and $g$ be a joint gain function from $\W \times \X$ to $[0 ,1]$. For $i\in\I$, let $g_i$ be a gain function from $\W_i  \times \X_i$ to $[0 ,1]$.
Given a prior $\pi$ on $\X$,
we define:
\begin{align*}
\Mminpi &= 
\min_{\overline{w} \in \W}\,
\min_{\overline{x} \in \S_{\overline{w}}}
\frac{ \prod_{i\in\I}\, \pi_i[x_i] g_i(w_i, x_i) }{ \pi[\overline{x}] \cdot g(\overline{w}, \overline{x}) }
\\[2.0ex]
\Mmaxpi &= 
\max_{\overline{w} \in \W}\,
\max_{\overline{x} \in \S_{\overline{w}}}
\frac{ \prod_{i\in\I}\, \pi_i[x_i] g_i(w_i, x_i) }{ \pi[\overline{x}] \cdot g(\overline{w}, \overline{x}) }
\end{align*}
where $\overline{w} = (w_1, w_2, \cdots, w_n)$ and $\overline{x} = (x_1, x_2, \cdots, x_n)$.
\end{defi}
\medskip

For the $n$-ary composition with distinct input we obtain the following extension of Theorem~\ref{thm:disjoint-dep-g-leak}
by using the above measures $\Mminpi$ and $\Mmaxpi$ for the $n$-ary composition.
\begin{prop}\label{thm:disjoint-dep-g-leak-k-channels}
For any jointly supported prior $\pi$ on $\X$ and channels $(\X_i, \Y_i, \allowbreak C_i)$ for $i\in\I$,
\begin{align*}
\displaystyle
\sum_{i\in\I} I_{g_i}(\pi_i, C_i) - \log\!\frac{ \Mmaxpi }{ \Mminpi }
\le
\Ig\Bigl(\!\pi, \prod_{i\in\I} C_i\!\Bigr)
\le
\sum_{i\in\I} I_{g_i}(\pi_i,\, C_i) + \log\!\frac{ \Mmaxpi }{ \Mminpi }.
\end{align*}
\end{prop}
%
\proof
The proof is by Theorem~\ref{thm:disjoint-dep-g-leak} and by induction on $\I$.
\qed

Theorem~\ref{thm:g-leak} also extends naturally to the $n$-ary shared input case:
\begin{prop} \label{thm:g-leak-n-channels}
For any prior $\pi$ on $\X$ and channels $(\X, \Y_i, C_i)$ for $i\in\I$,
\begin{align*}
\Ig(\pi, \compn_{i\in\I} C_i) \le
\sum_{i\in\I} I_{g}(\pi, C_i) + (\#\I - 1) \cdot \left( \HgMin(\pi) - \Hg(\pi) \right).
\end{align*}
\end{prop}
%
\proof
The proof is by Theorem~\ref{thm:g-leak} and by induction on $\I$.
\qed

%
%
\section{Compositionality Results on Min-Entropy Leakage}
\label{sec:compo-MEL}
In this section we present compositionality results for min-entropy leakage, which we recall is a particular case of $g$-leakage obtained when $g$ is the identity gain function.
The results here are instances of the results for $g$-leakage in the previous section while they yield compositionality theorems for min-capacity.
The summary of these results are shown in Table~\ref{table:compositionality:min-leak}.

\begin{table}[t]
\newcommand{\bhline}[1]{\noalign{\hrule height #1}}
\begin{center}
\begin{small}
  \def\arraystretch{1.1}
 \begin{tabular}{|l|c|c|c|}
 \hline
 Composition & Measure					& Lower bound		& Upper bound \\[0.4ex] \bhline{0.3mm}
 ~Composition	& posterior min-entropy
 	& Corollary~\ref{cor:disjoint-vul} (1)	& Corollary~\ref{cor:disjoint-vul} (2) \\[0.4ex]
 ~with distinct	&\hspace{3ex}\scriptsize{with independent priors}
 	& \multicolumn{2}{c|}{\hspace{-8ex}Corollary~\ref{cor:disjoint-vul} (3)} \\[0.4ex] \cline{2-4}
 ~input $\compd$& min-entropy leakage
 	& Corollary~\ref{cor:disjoint-min-entropy-leak} (1) & Corollary~\ref{cor:disjoint-min-entropy-leak} (2)  \\[0.4ex]
 &\hspace{3ex}\scriptsize{with independent priors}
 	& \multicolumn{2}{c|}{\hspace{-8ex}Corollary~\ref{cor:disjoint-min-entropy-leak} (3)} \\[0.4ex] \cline{2-4}
 			& min-capacity
	& \multicolumn{2}{c|}{Corollary~\ref{cor:disjoint-min-cap} \scriptsize{(first appeared in~\cite{barthe2011information})}} \\[0.4ex] \hline
 ~Composition  & posterior min-entropy
 	& Corollary~\ref{cor:vul} (1)	& Corollary~\ref{cor:vul} (2) \\[0.4ex] \cline{2-4}
 ~with shared & min-entropy leakage
 	& Corollary~\ref{cor:vul} (3)	& Corollary~\ref{cor:vul} (4)  \\[0.4ex] \cline{2-4}
 ~input $\compn$& min-capacity
 	& {---} & Corollary~\ref{cor:min-cap} \scriptsize{(first appeared in~\cite{espinoza2013min})}  \\[0.4ex] \hline
 \end{tabular}
 \caption{Compositionality results for min-entropy leakage.}
 \label{table:compositionality:min-leak}
\end{small}
\end{center}
\end{table}

\subsection{Min-Entropy Leakage of  Parallel Composition with Distinct Input}
\label{subsec:MEL-composition}
In this section we derive bounds for the min-entropy leakage of the parallel composition with distinct input.

We start by remarking that,  when the gain functions are identity gain functions,  $\Mminpi$ and $\Mmaxpi$ reduce to  
$\MminInfpi$ and $\MmaxInfpi$ defined as:
\begin{align*}
\MminInfpi &= 
\min_{(x_1, x_2) \in (\X_1\times \X_2)^+}
\frac{ \pi_1[x_1] \cdot \pi_2[x_2] }{ \pi[x_1, x_2] } \\[0.5ex]
\MmaxInfpi &= 
\max_{(x_1, x_2) \in (\X_1\times \X_2)^+}
\frac{ \pi_1[x_1] \cdot \pi_2[x_2] }{ \pi[x_1, x_2] }
\end{align*}
where $(\X_1\times \X_2)^+ = \{ (x_1, x_2)\in\X_1\times\X_2 \colon \pi[x_1,x_2]\neq 0 \}$.

The next results are consequences of the results of Section~\ref{sec:compo-g-leak}:

\begin{cor} \label{cor:disjoint-vul}
For any prior $\pi$ on $\X_1 \times \X_2$ and two channels $(\X_1, \Y_1, C_1)$, $(\X_2, \Y_2, C_2)$,
\begin{enumerate}
\item
$\Hinf(\pi, C_1 \compd C_2) \ge
\Hinf(\pi_1, C_1) + \Hinf(\pi_2, C_2) + \log \MminInfpi
$,
\item 
$
\Hinf(\pi, C_1 \compd C_2)\!\le\!
\Hinf(\pi_1, C_1) + \Hinf(\pi_2, C_2) + \log \MmaxInfpi
$
if $\pi$ is jointly supported,
\item if $\pi = \pi_1 \times \pi_2$, then
$\Hinf(\pi_1 \times \pi_2, C_1 \compd C_2) = \Hinf(\pi_1, C_1) + \Hinf(\pi_2, C_2)$.
\end{enumerate}
\end{cor}
%
\proof
By taking $g$, $g_1$, $g_2$ as the identity gain functions, 
$H_{g}(\pi_1 \times \pi_2, C_1 \compd C_2) = \Hinf(\pi_1 \times \pi_2, C_1 \compd C_2)$,
$H_{g_1}(\pi_1, C_1) = \Hinf(\pi_1, C_1)$, $H_{g_2}(\pi_2, C_2) = \Hinf(\pi_2, C_2)$,
$\Mminpi = \Mmin{\infty, \pi}$ and $\Mmaxpi = \Mmax{\infty, \pi}$.
Then the first and second claims follow from (1) and (2) of Lemma~\ref{lem:disjoint-dep-g-vul} respectively.
If $\pi_1$ and $\pi_2$ are independent we have $\Mmin{\infty, \pi} = \Mmax{\infty, \pi} = 1$, hence the last claim follows.
\qed

%
\begin{cor} \label{cor:disjoint-min-entropy-leak}
For any jointly supported prior  $\pi$ on $\X_1 \times \X_2$ and two channels $(\X_1, \Y_1, C_1)$, $(\X_2, \Y_2, C_2)$,
\begin{enumerate}
\item 
$\Iinf(\pi, C_1 \compd C_2) \ge 
\Iinf(\pi_1, C_1)\!+\!\Iinf(\pi_2, C_2)\,\!-\!\,\log\!\frac{ \MmaxInfpi }{ \MminInfpi }$,
\item 
$\Iinf(\pi, C_1 \compd C_2)%
\,\!\le\!\,
\Iinf(\pi_1, C_1)\!+\!\Iinf(\pi_2, C_2)\,\!+\!\,\log\!\frac{ \MmaxInfpi }{ \MminInfpi }$,
\item if $\pi = \pi_1 \times \pi_2$, then
$\Iinf(\pi_1 \times \pi_2, C_1 \compd C_2) = \Iinf(\pi_1, C_1) + \Iinf(\pi_2, C_2)$.
\end{enumerate}

\end{cor}
%
\proof
By taking $g$, $g_1$, $g_2$ as the identity gain functions, 
$I_{g}(\pi_1 \times \pi_2, C_1 \compd C_2) = \Iinf(\pi_1 \times \pi_2, C_1 \compd C_2)$,
$I_{g_1}(\pi_1, C_1) = \Iinf(\pi_1, C_1)$, $I_{g_2}(\pi_2, C_2) = \Iinf(\pi_2, C_2)$,
$\Mminpi = \Mmin{\infty,\pi}$ and $\Mmaxpi = \Mmax{\infty,\pi}$.
Then the first and second claims follow from Theorem~\ref{thm:disjoint-dep-g-leak} and the last from Corollary~\ref{cor:disjoint-g-leak} respectively.
\qed

The min-entropy leakage coincides with the min-capacity when the prior $\pi$ is uniform. Thus we re-obtain the following result from the literature~\cite{barthe2011information}:
\begin{cor} \label{cor:disjoint-min-cap}
For any two channels $(\X_1, \Y_1, C_1)$ and $(\X_2, \Y_2, C_2)$,
\[
\Cinf(C_1 \compd C_2) = \Cinf(C_1) + \Cinf(C_2).
\]
\end{cor}
%
\proof
Let $\pi$ be the uniform distribution on $\X_1 \times \X_2$. Then $\pi = \pi_1 \times \pi_2$.
Hence $\Cinf(C_1 \compd C_2) = \Iinf(\pi, C_1 \compd C_2) = \Iinf(\pi_1, C_1) + \Iinf(\pi_2, C_2)$ by Corollary~\ref{cor:disjoint-min-entropy-leak}.
Since $\pi_1$ and $\pi_2$ are also uniformly distributed, we obtain $\Iinf(\pi_1, C_1) = \Cinf(C_1)$ and $\Iinf(\pi_2, C_2) = \Cinf(C_2)$.
Therefore $\Cinf(C_1 \compd C_2) = \Cinf(C_1) + \Cinf(C_2)$.
\qed
%
Note that~\cite{barthe2011information} does \emph{not} present any compositionality result for min-entropy leakage while proving the result for min-capacity in a different way.

\subsection{Min-Entropy Leakage of Parallel Composition with Shared Input}
\label{subsec:MEL-shared-composition}
In this section we show bounds on the min-entropy leakage of the parallel composition with shared input.
As corollaries of Theorems~\ref{thm:g-vul} and~\ref{thm:g-leak} we obtain the compositionality results for the posterior min-entropy and the min-entropy leakage by taking $g$ as the identity gain function $g_{id}$.
Then for any prior $\pi$ on $\X$,\, $\HgMin(\pi)$ is instantiated to the following:
\begin{align*}
\HMin(\pi) = -\log \min \left\{ \pi[x] \mid x \in \X, \, \pi[x] \neq 0 \right\}.
\end{align*}
Then $\HMin(\pi) \ge \log |\X| \ge \Hinf(\pi)$.

\begin{cor} \label{cor:vul}
For any prior $\pi$ on $\X$ and channels $(\X, \Y_1, C_1)$ and $(\X, \Y_2, C_2)$,
\begin{enumerate}
\item 
$\Hinf(\pi,C_1 \compn C_2) \ge \Hinf(\pi,C_1) + \Hinf(\pi,C_2) - \HMin(\pi)$,
\item 
$\Hinf(\pi,C_1 \compn C_2) \le \min\{ \Hinf(\pi,C_1), \Hinf(\pi,C_2) \}$,
\item
$\Iinf(\pi, C_1 \compn C_2) \ge \max\{ \Iinf(\pi, C_1), \Iinf(\pi, C_2) \}$,
\item
$\Iinf(\pi, C_1 \compn C_2) \le \Iinf(\pi, C_1) + \Iinf(\pi, C_2) + \HMin(\pi) - \Hinf(\pi)$.
\end{enumerate}
\end{cor}
%
\proof
When $w = x$, $\pi[x]g_{id}(w, x) = \pi[x]$ holds .
Hence $H_{g_{id}}^{min}(X) = \HMin(X)$.
By definition,  $H_{g_{id}}(\pi, C_1 \compn C_2) = \Hinf(\pi, C_1 \compn C_2)$,
$H_{g_{id}}(\pi, C_1) = \Hinf(\pi, C_1)$, $H_{g_{id}}(\pi, C_2) = \Hinf(\pi, C_2)$
and $H_{g_{id}}(\pi) = \Hinf(\pi)$.
Then the claims follow from Theorems~\ref{thm:g-vul} and~\ref{thm:g-leak}.
\qed

The min-entropy leakage coincides with the min-capacity when the prior $\pi$ is uniform.
If $\pi$ is uniform we have $\HMin(\pi) = \Hinf(\pi)$. Thus we re-obtain the following result from the literature~\cite{espinoza2013min}:
\begin{cor} \label{cor:min-cap}
For any two channels $(\X,\Y_{1},C_{1})$ and $(\X,\Y_{2},C_{2})$,
\[
\Cinf(C_{1} \compn C_{2}) \le \Cinf(C_{1}) + \Cinf(C_{2})
\texttt{.}
\]
\end{cor}

The following  is an example of the above inequality.
\begin{exa} \label{eg:min-cap} 
Consider the channel $(\X, \Y, C)$ shown in Example~\ref{eg:g-leak}.
Let $\pi$ be the uniform prior on $\X$.
Then $\Hinf(\pi) = 1$, $\Hinf(\pi, C) = \Hinf(\pi, C \compn C) \approx 0.152$.
Hence $\Cinf(C \compn C) = \Hinf(\pi) - \Hinf(\pi, C \compn C) \approx 0.848$ while 
$\Cinf(C) + \Cinf(C) \approx 1.696$.
\end{exa}

%
%
\section{Improving Leakage Bounds by Input Approximation}
\label{sec:input-approx}
The compositionality results shown in the previous sections may not give good bounds when the prior is far from the uniform distribution, as illustrated in Example~\ref{eg:g-leak}.
In particular, probabilities that are closer to $0$ in priors make our leakage bounds much worse.

Since such small probabilities do not affect true $g$-leakage values much, they can be removed from the priors while this may cause little error on $g$-leakage values.
In the following we present a way of improving bad $g$-leakage bounds by removing small probabilities from priors.
We call it \emph{input approximation}  technique. 
The idea of removing small probabilities is reminiscent of the notion of \emph{smooth entropy}~\cite{Cachin97},
although the motivation and technicalities are different.
For simplicity we will only show the case of min-entropy leakage, i.e., when $g$ is the identity gain function.
The summary of the results are shown in Table~\ref{table:compositionality:min-leak-approx}.

\begin{table}[t]
\newcommand{\bhline}[1]{\noalign{\hrule height #1}}
\begin{center}
\begin{small}
  \def\arraystretch{1.1}
 \begin{tabular}{|l|c|c|c|}
 \hline
 Composition & Measure				& Lower bound		& Upper bound \\[0.4ex] \bhline{0.3mm}
 ~with distinct	& min-entropy leakage
 	& Theorem~\ref{thm:comp-approx-MEL-joint-input}& Theorem~\ref{thm:comp-approx-MEL-joint-input} \\[0.4ex] \cline{3-4}
 ~input $\compd$&\hspace{11ex}\scriptsize{$n$-ary composition}
 	& Proposition~\ref{cor:comp-approx-MEL-joint-input-n-channels} & Proposition~\ref{cor:comp-approx-MEL-joint-input-n-channels}
 \\[0.4ex] \hline
 ~with shared & min-entropy leakage
 	& {---} & Theorem~\ref{thm:comp-approx-MEL-shared-input}  \\[0.4ex] \cline{3-4}
 ~input $\compn$&\hspace{11ex}\scriptsize{$n$-ary composition}
 	& {---} & Proposition~\ref{cor:comp-approx-MEL-shared-input-n-channels}  \\[0.4ex] \hline
 \end{tabular}
 \caption{Compositionality results for min-entropy leakage of black-box systems obtained by input approximation.}
 \label{table:compositionality:min-leak-approx}
\end{small}
\end{center}
\end{table}

\subsection{Bounds for White-Box Channels}
\label{subsec:known-channel}
We first consider the case in which an analyst knows the channel components. 
Given a prior $\pi$ on $\X$ we construct a sub-probability distribution of $\pi$ from which small probabilities are removed as follows.
Let $\X'$ be a non-empty proper subset of $\X$ such that 
\[
\max_{x' \in \X'} \pi[x'] \le \min_{x \in \X\setminus\X'} \pi[x]
\texttt{.}
\]
Then $\!\displaystyle\max_{x \in \X\setminus\X'} \pi[x] = \max_{x \in \X} \pi[x]$.
Let $\displaystyle \epsilon =\hspace{-1ex}\sum_{x' \in \X'}\hspace{-1ex} \pi[x']$.
We define a function $\pi|_{\X\setminus \X'} \colon \X \rightarrow [0, 1]$ by:
\[
\pi|_{\X\setminus \X'}[x] = \left\{
\begin{array}{lll}
0 &\quad&\mbox{if $x \in \X'$}\\
\pi[x] &&\mbox{otherwise}
\end{array}
\right.
\]
Then $\pi|_{\X\setminus \X'}$ is a sub-probability distribution, as it does not sum up to $1$; i.e., $\displaystyle\sum_{x \in \X} \pi|_{\X\setminus \X'}[x]  \allowbreak < 1$.
However, the results in previous sections do not require $\pi$ to be a probability distribution, and neither do the definitions of entropy and leakage.  
Errors caused by the above input approximation are bounded as follows:
\begin{thm} \label{thm:approx-g-leak}
For any prior $\pi$ on $\X$, channel $(\X, \Y, C)$, $\epsilon \ge 0$ and the set $\X'$ defined as above,
\[
 \Iinf(\pi|_{\X\setminus \X'}, C) \le \Iinf(\pi, C) \le\!
 \Iinf(\pi|_{\X\setminus \X'}, C) + \log\left( 1 + \frac{\epsilon}{V(\pi|_{\X\setminus \X'}, C)} \right)
\texttt{.}
\]
\end{thm}
%
\proof
\begin{align*}
\phantom{=} V(\pi, C)
&= \displaystyle \sum_{y \in \Y} 
\max_{x \in \X} \pi[x] C[x, y] \\
&\le \displaystyle \sum_{y \in \Y} \biggl(
\max_{x \in \X\setminus\X'} \pi[x] C[x, y]
+ \sum_{x' \in \X'} \pi[x'] C[x', y] \biggr) \\
&= \displaystyle \sum_{y \in \Y} \biggl(
\max_{x \in \X\setminus\X'} \pi[x] C[x, y] \biggr)
+ \sum_{x' \in \X'} \biggl( \pi[x'] \sum_{y \in \Y} C[x', y] \biggr) \\
&= \displaystyle \sum_{y \in \Y} \biggl(
\max_{x \in \X\setminus\X'} \pi[x] C[x, y] \biggr)
+ \epsilon \\[0.5ex]
&= V(\pi|_{\X\setminus \X'}, C) + \epsilon.
\end{align*}
By $\displaystyle \max_{x \in \X\setminus\X'} \pi[x] = \max_{x \in \X} \pi[x]$,
we have $V(\pi|_{\X\setminus \X'}) = V(\pi)$.
Then
\begin{align*}
\Iinf(\pi, C)
&= \displaystyle \log\frac{V(\pi, C)}{V(\pi|_{\X\setminus \X'})} \\[0.0ex]
& \le \displaystyle \log\frac{V(\pi|_{\X\setminus \X'}, C) + \epsilon}{V(\pi|_{\X\setminus \X'})} \\[0.0ex]
& = \displaystyle \log\left( \frac{V(\pi|_{\X\setminus \X'}, C)}{V(\pi|_{\X\setminus \X'})} \cdot \frac{V(\pi|_{\X\setminus \X'}, C) + \epsilon}{V(\pi|_{\X\setminus \X'}, C)}\right) \\[0.0ex]
& = \displaystyle \Iinf(\pi|_{\X\setminus \X'}, C) + \log\left( 1 + \frac{\epsilon}{V(\pi|_{\X\setminus \X'}, C)}\right).
\end{align*}
On the other hand, by $V(\pi|_{\X\setminus \X'}, C) \le V(\pi, C)$,\, we obtain $\Iinf(\pi|_{\X\setminus \X'}, C) \le \Iinf(\pi, C)$.
\qed

So, the idea is to remove very small probabilities in priors and then apply our compositional approach to derive bounds illustrated in a previous section. 
This will allow us to obtain better bounds, as small probabilities affect dramatically the precision of our approach, while removing them produces only relatively small errors as shown in Theorem~\ref{thm:approx-g-leak}.

More precisely, the technique works as follows.
Consider a channel $C$ composed of $C_1$ and $C_2$ in parallel and a joint prior $\pi$ on $\X_1 \times \X_2$.
We take $\X_1 \times \X_2$ as $\X$ in the input approximation procedure and Theorem~\ref{thm:approx-g-leak}.
To apply our compositional approach, 
we  take a $\X' \subseteq \X_1 \times \X_2$ so that $\pi|_{\X\setminus \X'}$ is jointly supported.
Then we apply Corollary~\ref{cor:disjoint-min-entropy-leak} to obtain a lower and an upper bound for $\Iinf(\pi|_{\X\setminus \X'}, C)$. Finally we apply Theorem~\ref{thm:approx-g-leak} to obtain bounds for the original $\Iinf(\pi, C)$.

\begin{exa} \label{eg:input-approx:min-entropy} 
\begin{figure}[h]
\begin{rm}
\begin{tabular}{|l||c|c|c|} 
\hline
		& $y_0$ 	& $y_1$ 	& $y_2$ \\ \hline \hline
$x_0$	& 0.50	& 0.23	& 0.27 \\ \hline
$x_1$	& 0.20	& 0.40	& 0.40 \\ \hline
$x_2$	& 0.21	& 0.43	& 0.36 \\ \hline
\end{tabular}
\caption{Channel matrix.}
\label{fig:channel-matrix}
\end{rm}
\end{figure}
Consider the channel $(\X, \Y, C)$ for $\X = \{x_0, x_1, x_2\}$, $\Y = \{y_0, y_1, y_2\}$ 
and $C$ is given in Figure~\ref{fig:channel-matrix}.
We assume the prior $\pi$ such that $\pi(x_0) = 0.01$,  $\pi(x_1) = 0.49$ and $\pi(x_2) = 0.50$, is shared among channels.
Then the min-entropy leakage of the channel $C^{10}$ composed of ten $C$'s in parallel is $0.1319$, while our upper bound is $0.7444$ when $\epsilon = 0.01$.
On the other hand, the upper bound obtained  using min-capacity~\cite{barthe2011information}   is $4.114$, which is much larger than ours.
\end{exa}
\subsection{Bounds for Channels Composed of Black-Box Channels}
\label{subsec:unknown-channel}
In some situations an analyst may not know the channel matrices $C_1$, $C_2$ and therefore cannot calculate $\Iinf(\pi|_{\X\setminus \X'}, C_i)$ or $V(\pi|_{\X\setminus \X'}, C_i)$ (necessary to apply Corollary~\ref{cor:disjoint-min-entropy-leak}), while he may know the information leakages  $\Iinf(\pi_1, C_1)$ and $\Iinf(\pi_2, C_2)$.
Our input approximation technique allows us to obtain bounds also in this case, although less precise than in the case of known channels.
Hereafter we let $\pi' = \pi|_{\X\setminus \X'}$.
We obtain the following from Theorem~\ref{thm:approx-g-leak}:
\begin{thm} \label{thm:comp-approx-MEL-joint-input}
For any jointly supported prior $\pi$ on $\X_1 \compd \X_2$, two channels $(\X_1, \Y_1, C_1)$ and $(\X_2, \Y_2, C_2)$, $\epsilon > 0$ and the sub-distribution $\pi'$ defined as above, 
\begin{align*}
 \Iinf(\pi, C_1 \compd C_2)
 &\ge
 \Iinf(\pi_1, C_1) + \Iinf(\pi_2, C_2) \\
 &\hspace{2ex} - \log \frac{ \Mmax{\infty,\pi'}}{ \Mmin{\infty,\pi'} }
  - \log\!\frac{V(\pi_1, C_1)}{V(\pi_1, C_1) - \epsilon}
  - \log\!\frac{V(\pi_2, C_2)}{V(\pi_2, C_2) - \epsilon} \\[1ex]
 \Iinf(\pi, C_1 \compd C_2)
 &\le
  \Iinf(\pi_1, C_1) + \Iinf(\pi_2, C_2) \\
 &\hspace{2ex} + \log\!\frac{ \Mmax{\infty,\pi'}}{ \Mmin{\infty,\pi'} }
  + \log\!\frac{ \max( V(\pi_1, C_1), V(\pi_2, C_2) }{ \max( V(\pi_1, C_1), V(\pi_2, C_2) ) - \epsilon}
\texttt{.}
\end{align*}
\end{thm}

\begin{proof}
By the first claim in the proof for Theorem~\ref{thm:approx-g-leak},  
$V(\pi', C_1 \compd C_2) \ge V(\pi, C_1 \compd C_2) - \epsilon \ge
 \max( V(\pi_1, C_1), V(\pi_2, C_2) ) - \epsilon$.
Then
\begin{align*}
\Iinf(\pi, C_1 \compd C_2)
 \le\,& 
\Iinf(\pi', C_1 \compd C_2) +
{\textstyle \log\left( 1 + \frac{\epsilon}{V(\pi', C_1 \compd C_2)} \right)}
\\ \le\,& 
\Iinf(\pi', C_1 \compd C_2) +
{\textstyle \log\left( 1 + \frac{\epsilon}{ \max( V(\pi_1, C_1), V(\pi_2, C_2) ) - \epsilon} \right)}
\\ \le\,& 
\Iinf(\pi'_1, C_1) + \Iinf(\pi'_2, C_2) +
\log \frac{ \Mmax{\infty,\pi'} }{ \Mmin{\infty,\pi'} } +
{\textstyle \log\left( \frac{ \max( V(\pi_1, C_1), V(\pi_2, C_2) }{ \max( V(\pi_1, C_1), V(\pi_2, C_2) ) - \epsilon} \right)}
\\&&\hspace{-28.0ex}
\mbox{(By Corollary \ref{cor:disjoint-min-entropy-leak})}
\\ \le\,& 
\Iinf(\pi_1, C_1) + \Iinf(\pi_2, C_2) +
\log \frac{ \Mmax{\infty,\pi'} }{ \Mmin{\infty,\pi'} } +
{\textstyle \log\left( \frac{ \max( V(\pi_1, C_1), V(\pi_2, C_2) }{ \max( V(\pi_1, C_1), V(\pi_2, C_2) ) - \epsilon} \right)}
\\&&\hspace{-28.0ex}
\mbox{(By Theorem~\ref{thm:approx-g-leak})}
\end{align*}

The other inequality is shown as follows.
\begin{align*}
\Iinf(\pi, C_1 \compd C_2)
 \ge\, &
\Iinf(\pi', C_1 \compd C_2)
&\hspace{-31.0ex}\!\mbox{(By Theorem~\ref{thm:approx-g-leak})}
\\ \ge\, &
\Iinf(\pi'_1, C_1) + \Iinf(\pi'_2, C_2) -
\log\!\frac{ \Mmax{\infty,\pi'} }{ \Mmin{\infty,\pi'} }
&\hspace{-31.0ex}\!\mbox{(By Corollary~\ref{cor:disjoint-min-entropy-leak})}
\\ \ge\, & 
\Iinf(\pi_1, C_1) + \Iinf(\pi_2, C_2) -
\log\!\frac{ \Mmax{\infty,\pi'} }{ \Mmin{\infty,\pi'} }\! -
{\textstyle \log\!\left( \frac{V(\pi', C_1) + \epsilon}{V(\pi', C_1)} \right)}\! -
{\textstyle \log\!\left( \frac{V(\pi', C_2) + \epsilon}{V(\pi', C_2)} \right)}
\\&
&\hspace{-31.0ex}\!\mbox{(By Theorem~\ref{thm:approx-g-leak})}
\\ \ge\, & 
\Iinf(\pi_1, C_1) + \Iinf(\pi_2, C_2) -
\log\!\frac{ \Mmax{\infty,\pi'} }{ \Mmin{\infty,\pi'} }\! -
{\textstyle \log\!\left( \frac{V(\pi, C_1)}{V(\pi, C_1) - \epsilon} \right)}\! -
{\textstyle \log\!\left( \frac{V(\pi, C_2)}{V(\pi, C_2) - \epsilon} \right)} \tag*{\qEd}
\end{align*}
\def\popQED{}
\end{proof}

\begin{thm} \label{thm:comp-approx-MEL-shared-input}
For any prior $\pi$ on $\X$, two channels $(\X, \Y_1, C_1)$ and $(\X, \Y_2, C_2)$, $\epsilon > 0$ and the sub-distribution $\pi'$ defined as above,
\begin{align*}
\Iinf(\pi, C_1 \compn C_2)
&\le\, \Iinf(\pi_1, C_1) + \Iinf(\pi_2, C_2) \\[0.5ex]
&~~~~~+ \log\!\frac{ \max(V(\pi_1, C_1), V(\pi_2, C_2)) }{\max(V(\pi_1, C_1), V(\pi_2, C_2)) - \epsilon}
+ \HMin(\pi') - \Hinf(\pi')
\texttt{.}
\end{align*}
\end{thm}

\begin{proof}
Recall that $\pi' = \pi|_{\X\setminus \X'}$.
Then by Theorem~\ref{thm:approx-g-leak} and Corollary~\ref{cor:vul},
\begin{align*}
&
\Iinf(\pi, C_1 \compn C_2)
\\ \le\,& 
\Iinf(\pi', C_1 \compn C_2) +
{\textstyle
\log\left( 1 + \frac{\epsilon}{V(\pi', C_1 \compn C_2)} \right)}
\\ \le\,& 
\Iinf(\pi', C_1 \compn C_2) +
{\textstyle
\log\left( \frac{ \max( V(\pi_1, C_1), V(\pi_2, C_2) ) }{ \max( V(\pi_1, C_1), V(\pi_2, C_2) ) - \epsilon} \right)}
\\ \le\,& 
\Iinf(\pi'_1, C_1) + \Iinf(\pi'_2, C_2) +
\HMin(\pi') - \Hinf(\pi') +
{\textstyle
\log\!\left( \frac{ \max( V(\pi_1, C_1), V(\pi_2, C_2) }{ \max( V(\pi_1, C_1), V(\pi_2, C_2) ) - \epsilon} \right)}
\\ \le\,& 
\Iinf(\pi_1, C_1) + \Iinf(\pi_2, C_2) +
\HMin(\pi') - \Hinf(\pi') +
{\textstyle
\log\!\left( \frac{ \max( V(\pi_1, C_1), V(\pi_2, C_2) }{ \max( V(\pi_1, C_1), V(\pi_2, C_2) ) - \epsilon} \right)} \tag*{\qEd}
\end{align*}
\def\popQED{}
\end{proof}

In the right-hand side of these theorems $V(\pi_1, C_1)$ and $V(\pi_2, C_2)$ are calculated from $V(\pi_1)$, $V(\pi_2)$, $\Iinf(\pi_1, C_1)$ and $\Iinf(\pi_2, C_2)$.
So it is sufficient for an analyst to know only $\pi$, $\Iinf(\pi_1, C_1)$ and $\Iinf(\pi_2, C_2)$ to calculate the above leakage bounds.

It is easy to see that these bounds are not as good as those in Section~\ref{subsec:known-channel}. Also they are more sensitive to the choice of $\epsilon$.
If we take a very small $\epsilon$, the input approximation does not improve substantially, as neither $\frac{ \Mmax{\infty, \pi'} }{ \Mmin{\infty, \pi'} }$ nor $\HMin(\pi') - \Hinf(\pi')$ decreases much.
In fact, when $\epsilon = 0$ these theorems coincide with Corollaries~\ref{cor:disjoint-min-entropy-leak} and~\ref{cor:vul}.
If we take a very large $\epsilon$, then the error caused by the input approximation is also very large, while $\frac{ \Mmax{\infty, \pi'} }{ \Mmin{\infty, \pi'} }$ and $\HMin(\pi') - \Hinf(\pi')$ are close to $0$.
In Section~\ref{subsec:eval-rand} we will present experiments on the input approximation and illustrate that we should take $\epsilon$ as a value less than $\max\{ V(\pi_1, C_1), V(\pi_2, C_2) \}$.

\subsection{Extension to $n$-ary Composition}

The input approximation techniques illustrated in Sections~\ref{subsec:known-channel}  and \ref{subsec:unknown-channel} can be extended to $n$-ary composition.
The following results can be proven by induction on $\I$, where the inductive steps are similar to the proofs of Theorems~\ref{thm:comp-approx-MEL-joint-input} and~\ref{thm:comp-approx-MEL-shared-input}, respectively.
In the following $\X'$ is defined as in Section~\ref{subsec:known-channel}, and $\pi' = \pi|_{\X\setminus \X'}$.
\begin{prop} \label{cor:comp-approx-MEL-joint-input-n-channels}
For any channels $(\X_i, \Y_i, C_i)$ with $i\in\I$, $\X = \X_1 \times \ldots \times \X_n$ and any jointly supported prior $\pi$ on $\X$, 
\begin{align*}
 \Iinf\Bigl(\pi,\, \prod_{i\in\I} C_i\Bigr)
 &\ge \displaystyle
 \sum_{i\in\I} \left(
 \Iinf(\pi_i, C_i) -
 \log\Bigl( \frac{ V(\pi_i, C_i) }{V(\pi_i, C_i) - \epsilon} \Bigr) \right) -
 \log \frac{ \Mmax{\infty, \pi'} }{ \Mmin{\infty, \pi'} }
 \\[1.5ex]
 \Iinf\Bigl(\pi,\, \prod_{i\in\I} C_i\Bigr)
 &\le \displaystyle
 \biggl( \sum_{i\in\I}
 \Iinf(\pi_i, C_i) \biggr) +
 \log\!\left( \frac{ \displaystyle\max_{j\in\I}\, V(\pi_j, C_j) }{ \displaystyle\max_{j\in\I}\, V(\pi_j, C_j) - \epsilon} \right) +
  \log \frac{ \Mmax{\infty, \pi'} }{ \Mmin{\infty, \pi'} }
\texttt{.}
\end{align*}
\end{prop}
\begin{prop} \label{cor:comp-approx-MEL-shared-input-n-channels}
For any channels $(\X, \Y_i, C_i)$ with $i\in\I$ and any prior $\pi$ on $\X$, 
\begin{align*}
 \Iinf(\pi, \compn_{i\in\I} C_i)
  &\le
 \sum_{i\in\I}\!\Iinf(\pi_i, C_i) + (\#\I - 1)\!\cdot\!\Large(
 \HMin(\pi') -
 \Hinf(\pi') \Large) 
 + \log\!\Biggl(\!\frac{ \displaystyle\max_{j\in\I} V(\pi_j, C_j) }{\displaystyle\max_{j\in\I} V(\pi_j, C_j) - \epsilon}\!\Biggr).
\end{align*}
\end{prop}

%
%
\section{Leakage Bounds by Channel Approximation} \label{sec:channel-approx}
The leakage bounds can be computed more efficiently by using simpler channels that over/under-approximate the original complicated system.
In this section we show how to compute the leakage bounds using channels related by the cascade composition refinement.

\subsection{Comparing Composed Channels}
The comparison of channels from the viewpoint of information leakage has been investigated in several studies such as~\cite{YasuokaT10,Malacaria11,Alvim:12:CSF}.
The \emph{cascade composition refinement} $\sqsubseteq$ is shown to be useful to compare the $g$-leakages of two channels.
\begin{defi} \label{def:comparing-channels} \rm
For any two channels $(\X, \Z, C_1)$ and $(\X, \Y, C_2)$, we write $C_1 \sqsubseteq C_2$ if there exists a channel $(\Y, \Z, C_3)$ such that $C_1[x, z] = C_2[x, y] \cdot C_3[y, z]$ for any $x \in \X$, $y \in \Y$ and $z \in \Z$.
\end{defi}

\begin{rem}
Note that, for the relation $C_1 \sqsubseteq C_2$ to be defined, $C_1$ and $C_2$ need to have the same input domain, but not necessarily the same output domain.
\end{rem}

We obtain the following results.
\begin{prop} \label{prop:comparing-g-leak}
Let $C_1$, $C_2$, $C'_1$ and $C'_2$ be channel matrices.
\begin{enumerate}
\item 
If $C_1 \sqsubseteq C'_1$, then $C_1 \compd C_2 \sqsubseteq C'_1 \compd C_2$.
\item 
If $C_1 \sqsubseteq C'_1$ and $C_2 \sqsubseteq C'_2$, then $C_1 \compd C_2 \sqsubseteq C'_1 \compd C'_2$.
\item 
If $C_1 \sqsubseteq C'_1$, then $C_1 \compn C_2 \sqsubseteq C'_1 \compn C_2$.
\item 
If $C_1 \sqsubseteq C'_1$ and $C_2 \sqsubseteq C'_2$, then $C_1 \compn C_2 \sqsubseteq C'_1 \compn C'_2$.
\end{enumerate}
\end{prop}
%
  \begin{proof}
Immediate from Definitions~\ref{def:disjoint-compo},~\ref{def:parallel-compo}
and~\ref{def:comparing-channels}. \hfill \qEd

\vspace{1em}
The cascade composition refinement and the $g$-leakage ordering coincide: 
Alvim et al.~\cite{Alvim:12:CSF} show that $C_1 \sqsubseteq C_2$ implies $\Ig(\pi, C_1) \le \Ig(\pi, C_2)$ for all priors $\pi$ and gain functions $g$; 
McIver et al.~\cite{McIverMSEM14post} show the converse implication.
\def\popQED{}
\end{proof}

\subsection{Approximate Leakage Bounds by Channel Approximation}

In this section we show some refinement results that can be used to reason about lower\slash upper bounds on information leakage of huge channels.
These are straightforward from Proposition~\ref{prop:comparing-g-leak} and compositionality results in previous sections.

Given two complicated systems modeled as channels $(\X_1, \Y_1, C_1)$ and $(\X_2, \Y_2, C_2)$ respectively, we consider simpler systems that are described as channel matrices $C'_1$, $C'_2$, $C''_1$ and $C''_2$ such that $C'_1 \sqsubseteq C_1 \sqsubseteq C''_1$ and $C'_2 \sqsubseteq C_2 \sqsubseteq C''_2$.
The following result shows that an analyzer can learn a lower\slash upper bound on the $g$-leakage $\Ig(\pi, C_1 \compd C_2)$ from the knowledge/computation of $\IgA(\pi_1, C'_1)$ and $\IgB(\pi_2, C'_2)$ or that of $\IgA(\pi_1, C''_1)$ and $\IgB(\pi_2, C''_2)$, without calculating the $g$-leakages of the two original channels $C_1$ and $C_2$:

\begin{prop} \label{cor:compd-channel-approx}
For any jointly supported prior $\pi$ on $\X_1\times\X_2$,
\begin{enumerate}
\item
$\Ig(\pi, C_1 \compd C_2) \ge \IgA(\pi_1, C'_1) + \IgB(\pi_2, C'_2) - \log \frac{\Mmaxpi}{\Mminpi}$.
\item
$\Ig(\pi, C_1 \compd C_2) \le \IgA(\pi_1, C''_1) + \IgB(\pi_2, C''_2) + \log \frac{\Mmaxpi}{\Mminpi}$.
\end{enumerate}
\end{prop}
%
\proof
By~\cite{Alvim:12:CSF,McIverMSEM14post}, $C'_i \sqsubseteq C_i \sqsubseteq C''_i$ iff $\Ig(\pi_i, C'_i) \le \Ig(\pi_i, C_i) \le \Ig(\pi_i, C''_i)$ for each $i = 1, 2$.
Hence the claims are immediate from Proposition~\ref{prop:comparing-g-leak}, Theorem~\ref{thm:disjoint-dep-g-leak}.
\qed

Next we present a result for the case of shared input.
Given two channels $(\X, \Y_1, C_1)$ and $(\X, \Y_2, C_2)$ we consider two channel matrices $C'_1$ and $C'_2$ such that $C_1 \sqsubseteq C'_1$ and $C_2 \sqsubseteq C'_2$.
Then an analyzer can learn an upper bound on the $g$-leakage $\Ig(\pi, C_1 \compn C_2)$ from the knowledge on $\IgA(\pi_1, C'_1)$ and $\IgB(\pi_2, C'_2)$, without calculating the $g$-leakages of two original channels $C_1$ and $C_2$:

\begin{prop} \label{cor:compn-channel-approx}
For any jointly supported prior $\pi$ on $\X_1\times\X_2$,
\begin{align*}
\Ig(\pi, C_1 \compn C_2) \le 
\IgA(\pi_1, C'_1) + \IgB(\pi_2, C'_2) + \HgMin(\pi)- \Hg(\pi).
\end{align*}
\end{prop}
%
\proof
Immediate from Proposition~\ref{prop:comparing-g-leak} and Theorem~\ref{thm:g-leak}.
\qed

%
%
\section{Shannon Mutual Information}
\label{sec:compo-MI}
In this section we present compositionality results for Shannon \emph{mutual information} for the completeness of this paper.
These are easily obtained from Shannon information theory by applying it to our formalization of channel compositions.

Recall that, given a prior $\pi$ on $\X$ and a discrete channel $(\X, \Y, C)$, the mutual information is defined as:
\[
I(\pi, C) = \sum_{x \in \X, y \in \Y}
\pi[x] C[x, y] \log\left( \frac{ C[x, y] }{ \sum_{x' \in \X} \pi[x']C[x', y] } \right)
\texttt{.}
\]

Similarly to Corollary~\ref{cor:disjoint-g-leak} for $g$-leakage, the mutual information of channels composed in parallel is the summation of the mutual information of components when the priors are independent.
\begin{thm} \label{thm:MI-comp-disjoint}
For any two priors $\pi_1$ on $\X_1$, $\pi_2$ on $\X_2$ and any two discrete channels $(\X_1, \Y_1, C_1)$ and $(\X_2, \Y_2, C_2)$, we have 
\[
I(\pi_1 \times \pi_2, C_1 \compd C_2) = I(\pi_1, C_1) + I(\pi_2, C_2)
\texttt{.}
\]
\end{thm}
%
\proof
Immediate from the independency of the random variables.
\qed

In the case of shared input we have inequality like Theorem~\ref{thm:g-leak}.
\begin{thm} \label{thm:MI-comp}
For any prior $\pi$ on $\X$ and any two discrete channels $(\X, \Y_1, C_1)$ and $(\X, \Y_2, C_2)$, we have 
\[
I(\pi, C_1 \compn C_2) \le I(\pi, C_1) + I(\pi, C_2).
\]
\end{thm}
\begin{proof}
This is derived from a property of the \emph{interaction information} as follows.
Let $X$ be the random variable  for the input of the channels $C_1$ and $C_2$, and $Y_1$ and $Y_2$ be the two random variables for the outputs of the two channels.
Then
\[
\begin{array}[b]{rl}
&
I(\pi, C_1 \compn C_2) - I(\pi, C_1) - I(\pi, C_2)
\\ =&
I(X ; Y_1 Y_2) - I(X ; Y_1) - I(X ; Y_2)
\\ =&
I(X ; Y_1 | Y_2) + I(X ; Y_2) - I(X ; Y_1) - I(X ; Y_2)
\\ =&
I(X ; Y_1 | Y_2) - I(X ; Y_1)
\\ =&
I(Y_1 ; Y_2 | X) - I(Y_1 ; Y_2)
\\ =&
- I(Y_1 ; Y_2)
~~~~~~~~~~~~~~~~~~~\hfil\mbox{(by conditional independence)}
\\ \le&
0.
\end{array}\tag*{\qEd}
\]
\def\popQED{}
\end{proof}

These results can be extended to the $n$-ary compositions:
\begin{thm} \label{thm:MI-comp-disjoint-n}
For any priors $\pi_i$ on $\X_i$ and any discrete channels $(\X_i, \Y_i, C_i)$ for $i\in\I$, we have 
\[
I\Bigl( \prod_{i\in\I} \pi_i,\, \prod_{i\in\I} C_i \Bigr) = \sum_{i\in\I} I(\pi_i, C_i)
\texttt{.}
\]
\end{thm}
%
\proof
By Theorem~\ref{thm:MI-comp-disjoint} and by induction on $\I$.
\qed

%
\begin{thm} \label{thm:MI-comp-n}
For any prior $\pi$ on $\X$ and any discrete channels $(\X, \Y_i, C_i)$ for $i\in\I$, we have 
\[
I\large(\pi, \compn_{i\in\I} C_i\large) \le \sum_{i\in\I} I(\pi, C_i)
\texttt{.}
\]
\end{thm}
%
\proof
By Theorem~\ref{thm:MI-comp} and by induction on $\I$.
\qed

On the other hand, such a compositionality may not hold when the composition is adaptive, i.e., when the channels depend on each other.
\begin{prop} \label{prop:MI-countereg}
There exist a prior $\pi$ on $\X$ and a discrete channel $(\X, \Y_1 \times \Y_2, C)$ with two outputs such that
$
I(\pi, C) > I(\pi, C|_{\Y_1}) + I(\pi, C|_{\Y_2})
$.
\end{prop}
%
\proof
For example, let $\X = \Y_1 = \Y_2 = \{0, 1\}$ and $\pi$ be the uniform distribution on $\X$.
We consider the channel that, given an input $x \in \X$, outputs a random bit $y_1$ uniformly drawn from $\Y_1$ and the exclusive OR of the two bits $x$ and $y_1$ as the second output, i.e, $y_2 \mathbin{:=} x \xor y_1$.

Formally, the channel matrix $C$ is given by:
$C[0, (0, 0)] = C[0, (1, 1)] = 0.5$, $C[0, (0, 1)] = C[0, (1, 0)] = 0$,
$C[1, (0, 0)] = C[1, (1, 1)] = 0$ and $C[1, (0, 1)] = C[1, (1, 0)] = 0.5$.
Recall the definition of decomposition of channels in Definition~\ref{def:de-compo-o}.
Then, for all $x \in \X$, $y_1 \in \Y_1$ and $y_2 \in \Y_2$, $C|_{\Y_1}[x, y_1] = C|_{\Y_2}[x, y_2] = 0.5$.

The channel $C$ leaks the bit $x$ from $y_1$ and $y_2$ by $x \mathbin{:=} y_1 \xor y_2$, while neither $C|_{\Y_1}$ nor $C|_{\Y_2}$ leaks any information on $x$; i.e., $I(\pi, C) = 1$ and $I(\pi, C|_{\Y_1}) = I(\pi, C|_{\Y_2}) = 0$.
Therefore $I(\pi, C) > I(\pi, C|_{\Y_1}) + I(\pi, C|_{\Y_2})$.
\qed

\begin{rem}\rm\label{rem:shannon-as-g-leak}
As shown in~\cite{AlvimCMMPS14csf}, Shannon entropy is a particular instance of $g$-vulnerability.
However, due to the presence of the logarithm that we consider here, the compositionality of Shannon mutual information is not an instance of that of $g$-leakage.
\end{rem}

%
%
\section{Experimental Evaluation}
\label{sec:experiment}

In this section we evaluate our bounds in two use-cases:
first, on the Crowds protocol for anonymous communication, running on a mobile ad-hoc network
(MANET), and second, on randomly generated channels.

\subsection{Crowds Protocol on a MANET}
\label{subsec:crowds}

Crowds~\cite{Reiter:98:TISS} is a protocol for anonymous communication,
in which participants achieve anonymity by forwarding messages through other
users. A group of $n$ users, called the Crowd, participate in the protocol, and
one of them, called the \emph{initiator} decides to send a message to some
arbitrary recipient in the network, called the \emph{server}. The protocol
works as follows: first the initiator selects randomly (with uniform
distribution) a member of the crowd, called the \emph{forwarder}, and forwards
the message to him. A forwarder, upon receiving a message, throws a (biased)
probabilistic coin: with probability $p_f$ (a parameter of the system) he
randomly selects a new forwarder and advances the message to him, and with
probability $1-p_f$ he delivers the message directly to the server. Replies
from the server follow the inverse path to arrive to the initiator and future
requests use the already established route, to avoid repeating the protocol.

The goal of the protocol is to provide sender anonymity w.r.t. an attacker who
does not control the whole network, but controls only some of the nodes
and can only see traffic passing through them. Still, if the attacker controls
some members of the crowd, strong anonymity is not satisfied. A forwarding
request from user $i$ is evidence that $i$ is the initiator of the message.
However, some anonymity is still provided since user $i$ can always claim that
he was in fact only forwarding a message from user $j$. 
If the number of corrupted users is relatively small, 
it is more likely that $i$ is
innocent (i.e. the initiator is user $j \neq i$) than guilty, offering a
notion of anonymity called \emph{probable innocence} \cite{Reiter:98:TISS}.

In this section we consider an instance of Crowds running on a mobile ad-hoc
network, in which users are mobile and can communicate only to neighbouring
nodes hence the network topology changes frequently. Due to the network
changes, routes become invalid and the initiator needs to rerun the
protocol to establish a new route, which causes further information leakage.
Our goal is to measure how quickly the leakage increases as a function of the
number of re-executions. Concerning the attacker model, we assume that the
attacker (i) knows the network topology (this could be achieved using known
protocols for MANETs, e.g.~\cite{nassu2007topology}), (ii) controls some
members of the crowd and (iii) controls the server. For a given network
topology, the system is modeled by a channel with inputs $\mathrm{init}_i$, meaning
that user $i$ is the initiator. The observable events are $\mathrm{forw}_{j,k}$, meaning
that user $j$ forwarded the message to the corrupted node $k$ (possibly
the destination server). A matrix element $C[\mathrm{init}_i,\mathrm{forw}_{j,k}]$ gives the probability
that $\mathrm{forw}_{j,k}$ happens when $i$ is the initiator.
Finally, for channels $C_1,C_2$ modeling the protocol under different network topologies, the repetition of the protocol is modeled as $C_1 \compn C_2$.

As an anonymity metric, we use $g$-leakage with the $2$-tries gain function $g_{\W_2}$, modeling an attacker who can guess the initiator twice.
Recall that $\W_2$ is the set of all subsets of $\X$ with $\#{\X} = 2$,
and $g_{\W_2}(w, x)$ is $1$ if $x \in w$ and $0$ otherwise.

We evaluate our compositionality results on a Crowds instance with $25$ users,
of which one is corrupted, and with $p_f=0.7$. The network topology is
generated by randomly adding a connection between any two users with
probability $0.4$. For a given topology, the matrix is computed by the PRISM
model checker~\cite{Kwiatkowska:04:QEST}, using a model similar to one of
\cite{Shmatikov:02:CSFW}. Although executions in Crowds can be infinite, a
finite state model can be employed, keeping track of only the current forwarder
instead of the full route. Then each element of the channel matrix can be
computed by PRISM as the probability of reaching the corresponding state.

%
\begin{wrapfigure}[17]{r}{0.56\linewidth}
  \begin{center}
  \includegraphics[scale=0.33]{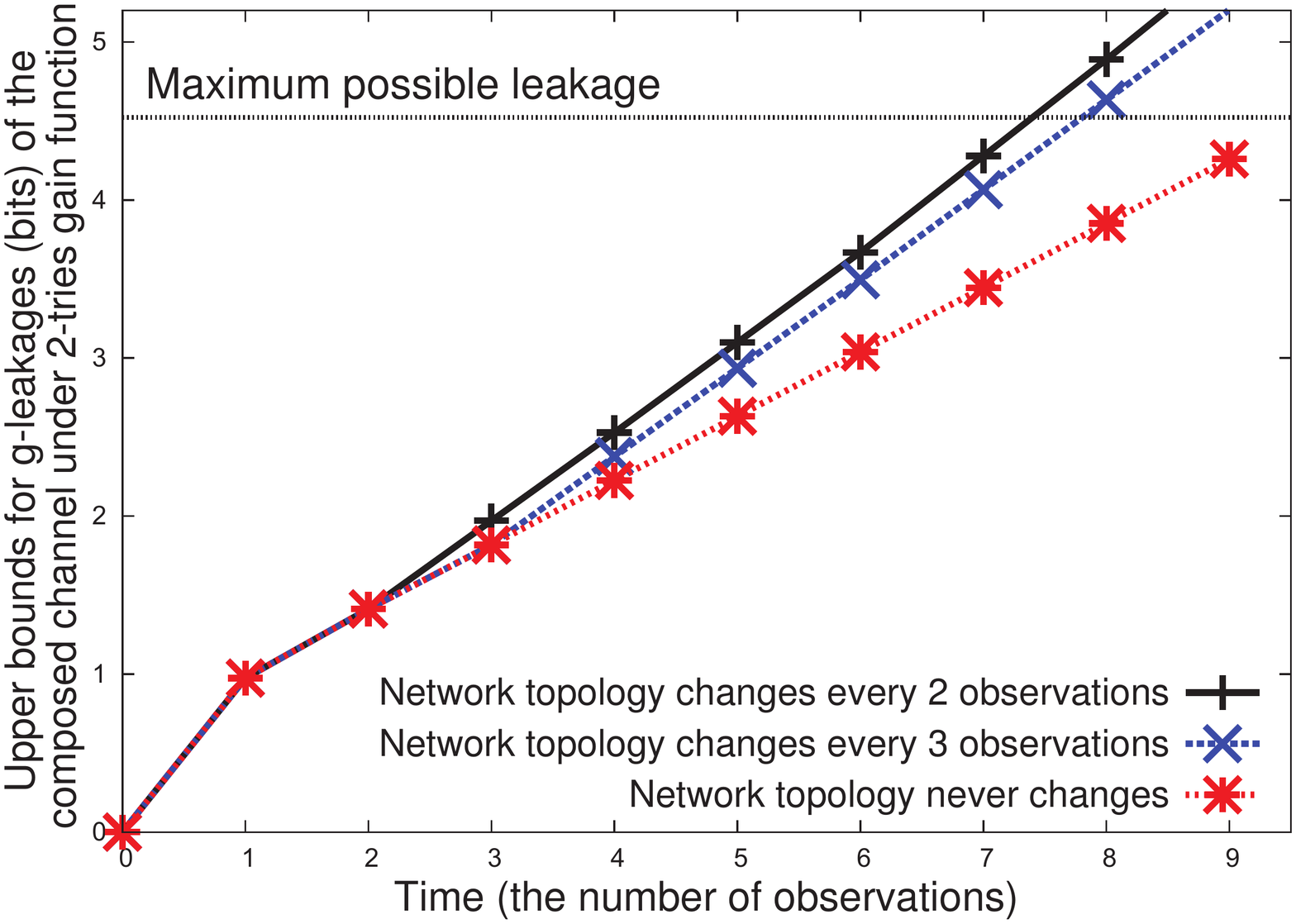}   
  \caption{Numbers of observations and bounds.} \label{fig:crowds1}
  \end{center}
\end{wrapfigure}
%
The $g$-leakage of a single execution can be directly computed from the
channel; however, 
for multiple executions, the channel quickly
becomes too big to be of practical use (already at 5 repetitions). On the other hand,
$g$-leakage can be bounded using the results in
Section~\ref{sec:compo-MEL}.
The obtained bounds for up to 9 protocol repetitions are shown in Figure~\ref{fig:crowds1}.
Three variations are given in which the topology changes every 2 executions, every 3 executions or always stays the same.
All bounds are computed using a uniform prior and some randomly generated channels.
The experiments show that the compositionality technique allows us to obtain meaningful
bounds when the system is too big to compute exact values.

Note that the assumption of uniformly chosen forwarders is standard for the Crowds protocol, however it would be interesting to study how our results would change if we considered non-uniform distributions.
For instance, we could have a non-uniform distribution if the possible forwarders were equipped with a notion of trust, like in~\cite{SassoneHY10}.
We leave this for future work.

%
%
%
\subsection{Evaluation on Randomly Generated Channels}
\label{subsec:eval-rand}

In this section we evaluate our bounds on min-entropy leakage
using randomly generated channels. In particular, we evaluate the improvement 
on the bounds due to the input approximation technique, and the efficiency of our approach, 
which we have implemented as
a library of \leakiEst{} version 1.3~\cite{Kawamoto2014tool}
that is an extension of \leakiEst~\cite{chothia2013tool} and is available online.

We first compare the exact leakage values with their upper bounds calculated using the input approximation technique in the case of shared input.
Figure~\ref{fig:known-eps0.1} shows the average upper bounds obtained from Theorem~\ref{thm:approx-g-leak}, that can be applied  when we know the channel matrix. Figure~\ref{fig:unknown-epsOpt} shows those obtained from Theorem~\ref{thm:comp-approx-MEL-shared-input} 
that we can apply when we \emph{do not  know} it.
For both experiments we use randomly generated $10 \times 10$ channel matrices $C$ and a prior $\pi$ that contains some input with very small probabilities. We set   
$\epsilon = 0.1$ in the first case and $\epsilon = V(\pi, C) / 3$ in the second one. 
We calculated the min-entropy leakage $\Iinf(\pi, C \compn C \compn C \compn C \compn C \compn C)$ (composition of six $C$'s), and its lower and upper bounds, 
using Corollary~\ref{cor:vul}, Theorem~\ref{thm:approx-g-leak} with the analyst's knowledge about $C$ and Theorem~\ref{cor:comp-approx-MEL-shared-input-n-channels} without the analyst's knowledge about $C$.

These cases give  similar upper bounds as shown in Figures~\ref{fig:known-eps0.1} and~\ref{fig:unknown-epsOpt}.
The x-axis represents \emph{noise levels} of randomly generated matrices, which we define as the maximum values (over rows of $C$) of the summations of the differences of probabilities from the uniform distributions.
For instance, when the noise level is $0.10$, the average upper bound is $1.699$ in the first case (Figure~\ref{fig:known-eps0.1}) while it is $1.701$ in the second (Figure~\ref{fig:unknown-epsOpt}).

\begin{figure}[t]
 \begin{minipage}{0.45\hsize}
  \begin{center}
  \includegraphics[scale=0.28]{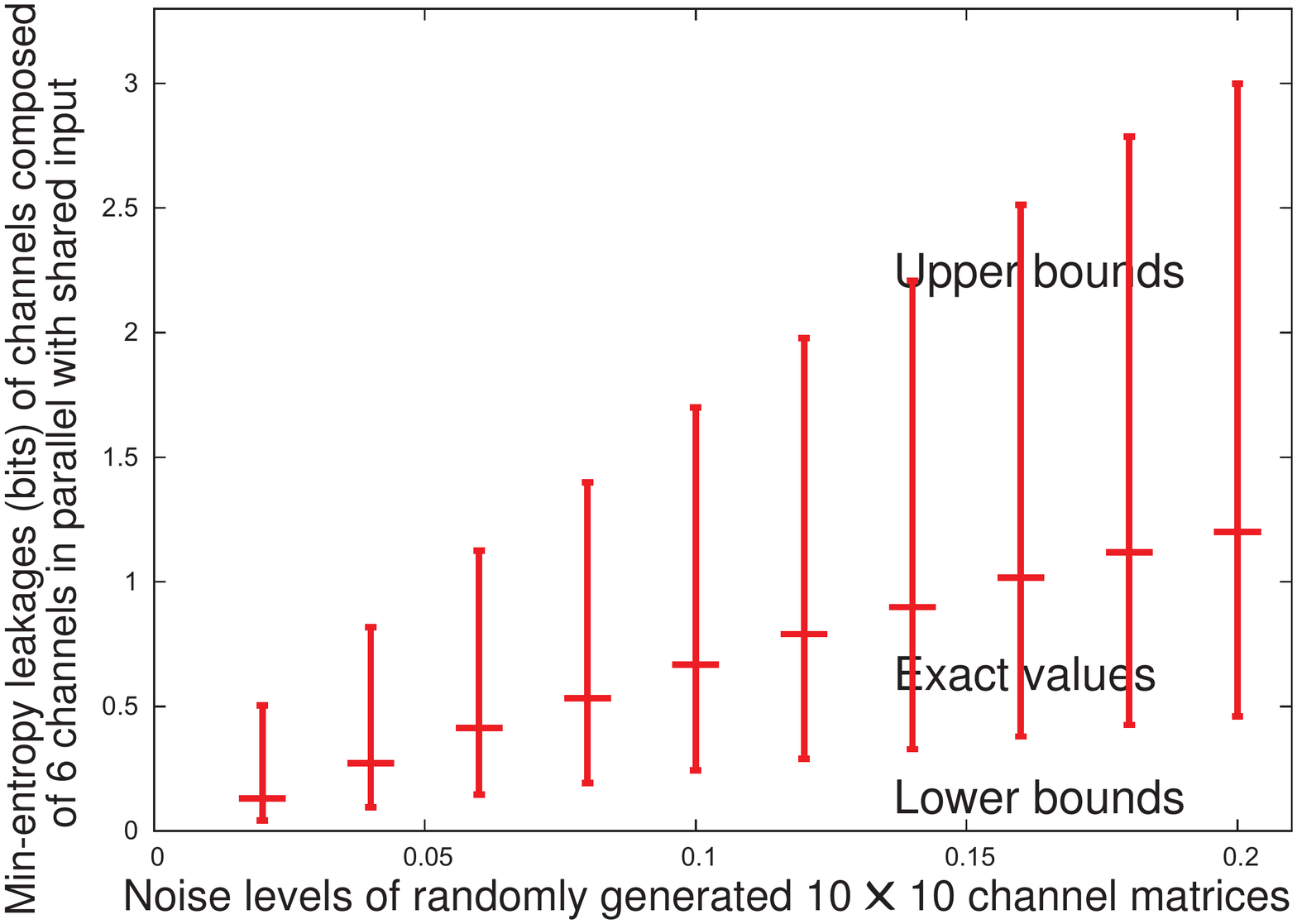}
  \caption{Min-entropy leakages and their bounds 
  for \emph{known} channels
  }
  \label{fig:known-eps0.1}
  \end{center}
 \end{minipage}
 \begin{minipage}{0.05\hsize}~
 \end{minipage}
 \begin{minipage}{0.45\hsize}
  \begin{center}
  \includegraphics[scale=0.28]{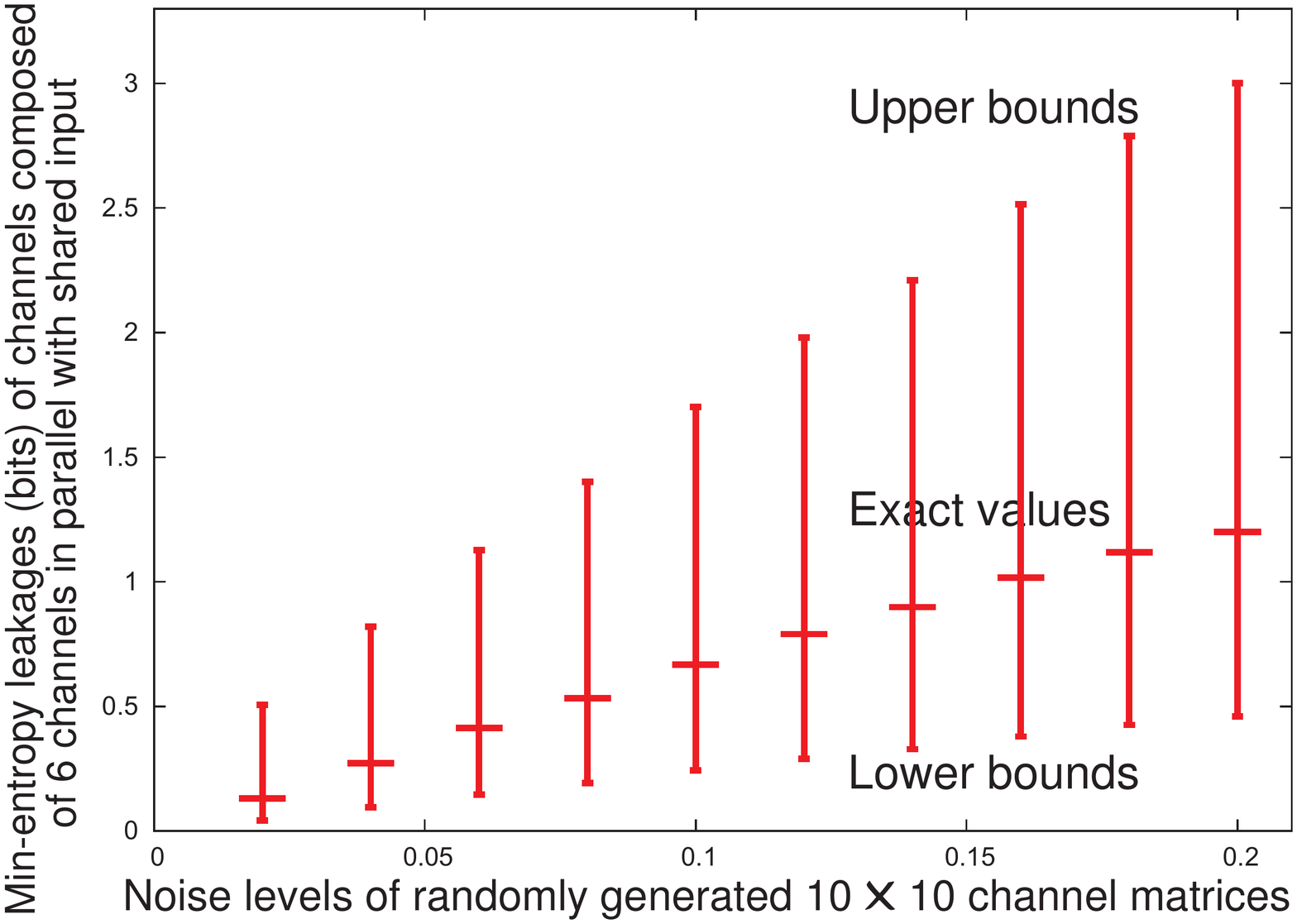}
  \caption{Min-entropy leakages and their bounds 
  for  \emph{unknown} channels.
  }
  \label{fig:unknown-epsOpt}
  \end{center}
 \end{minipage}
\end{figure}

These upper bounds depend on how we choose the parameter $\epsilon$ 
for the input approximation technique.
In particular upper bounds strongly depend on $\epsilon$ in the  case of unknown channels.
In Figure~\ref{fig:unknown-epsOpt} we chose $\epsilon = V(\pi, C) / 3$ which gives a relatively good upper bound.
On the other hand, if we choose an $\epsilon$ too large we may obtain useless bounds. 
Indeed, if we set for instance $\epsilon = 0.2$, then we obtain upper bounds 
above the maximum possible leakage, which is the (prior) min-entropy, and is always $\log 10 \approx 3.322$ (as shown in Figure~\ref{fig:unknown-eps0.2}) since the input is shared.

\begin{figure}[t]
 \begin{minipage}{0.45\hsize}
  \begin{center}
  \includegraphics[scale=0.28]{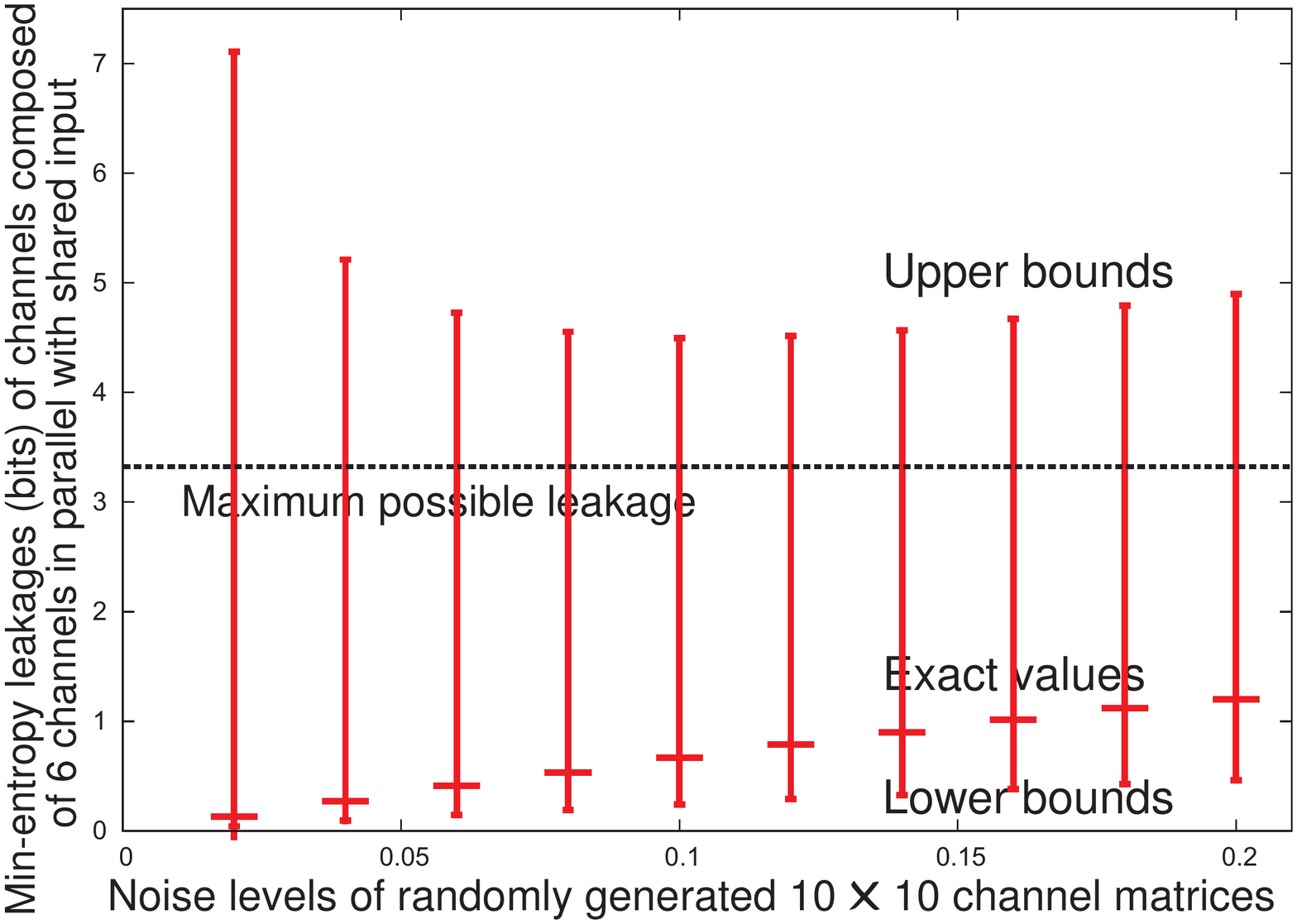}
  \caption{Min-entropy leakages and their bounds 
  when the analyst does \emph{not} know the channels and chooses $\epsilon$ poorly}
  \label{fig:unknown-eps0.2}
  \vspace{-1.5ex}
  \end{center}
 \end{minipage}
 \begin{minipage}{0.05\hsize}~
 \end{minipage}
 \begin{minipage}{0.45\hsize}
  \begin{center}
  \includegraphics[scale=0.28]{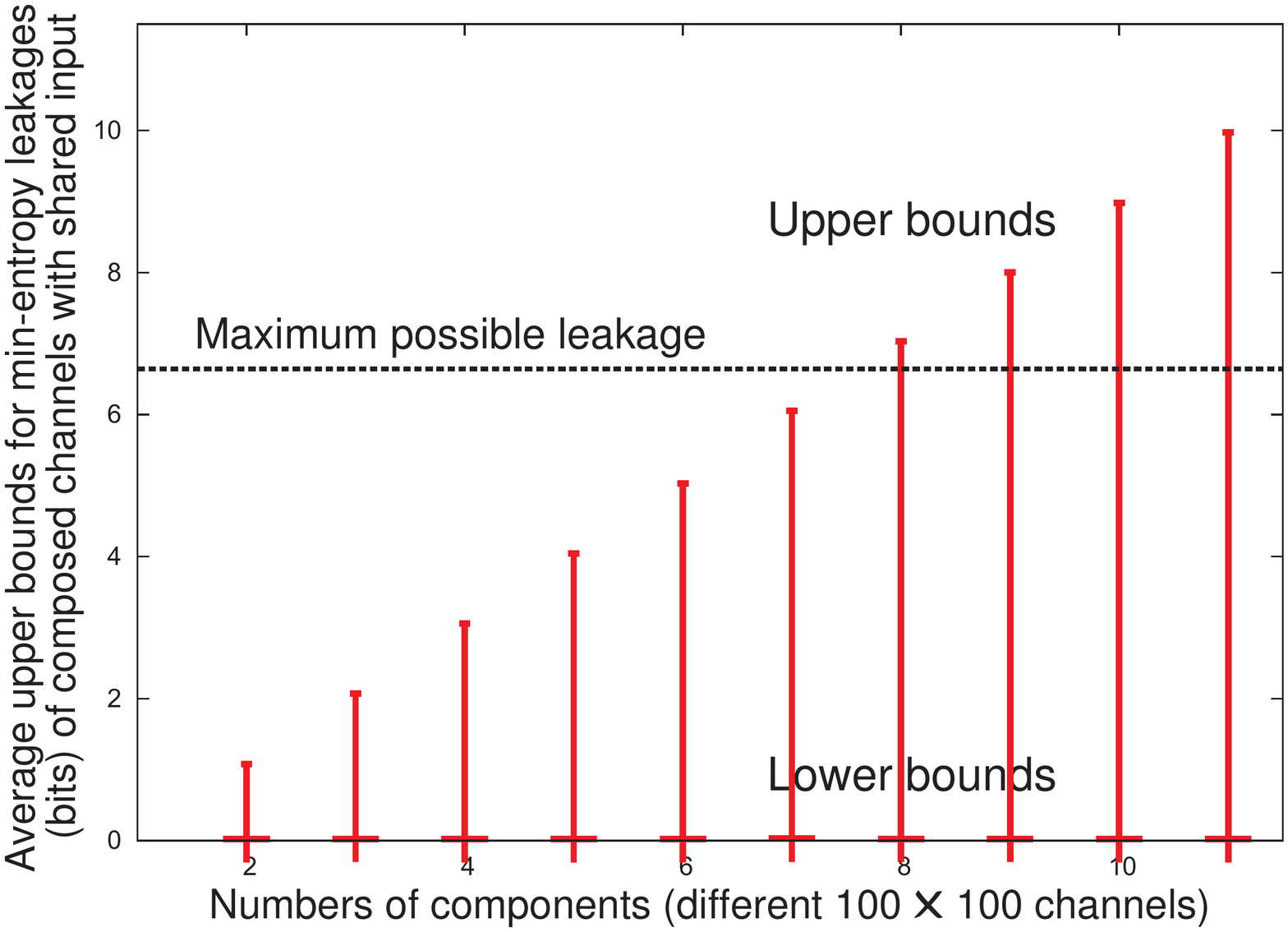}
  \caption{Upper bounds on min-entropy leakages as a function of the numbers of components.}
  \label{fig:100-100-known-eps0.005}
  \end{center}
 \end{minipage}
\end{figure}

Figure~\ref{fig:100-100-known-eps0.005} shows average upper bounds on min-entropy leakages of randomly generated $100 \times 100$ channels, with randomly generated priors, noise level $0.1$, and $\epsilon = 0.005$.
As we can see from the figure, the gap between the lower and upper bounds increases with the number of components.

\begin{figure}[t]
 \begin{minipage}{0.45\hsize}
  \begin{center}
  \includegraphics[scale=0.28]{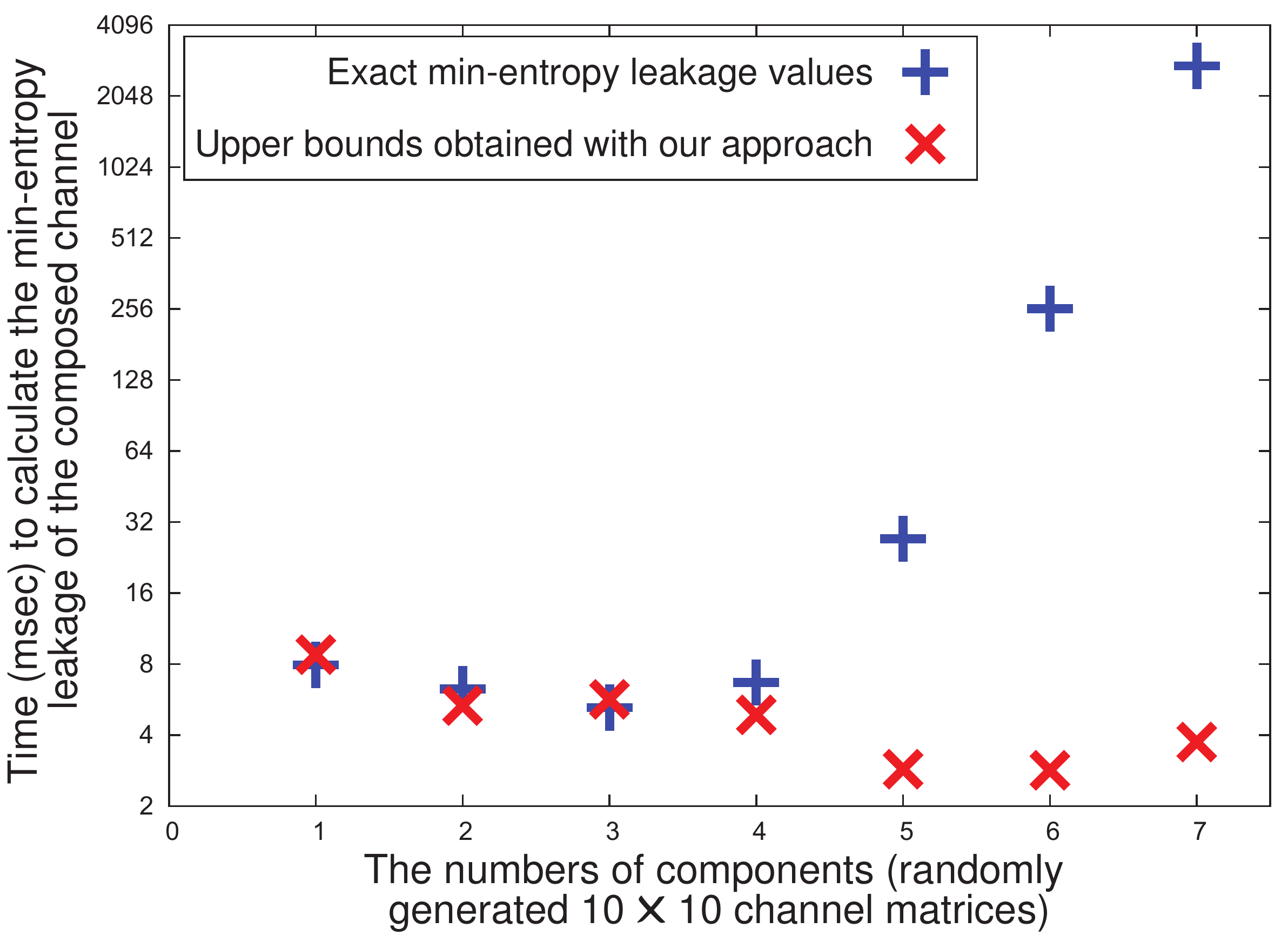}
  \caption{Average time to calculate min-entropy leakages and their upper bounds}
  \label{fig:time-comparison}
  \vspace{-2ex}
  \end{center}
 \end{minipage}
 \begin{minipage}{0.05\hsize}~
 \end{minipage}
 \begin{minipage}{0.45\hsize}
  \begin{center}
  \includegraphics[scale=0.27]{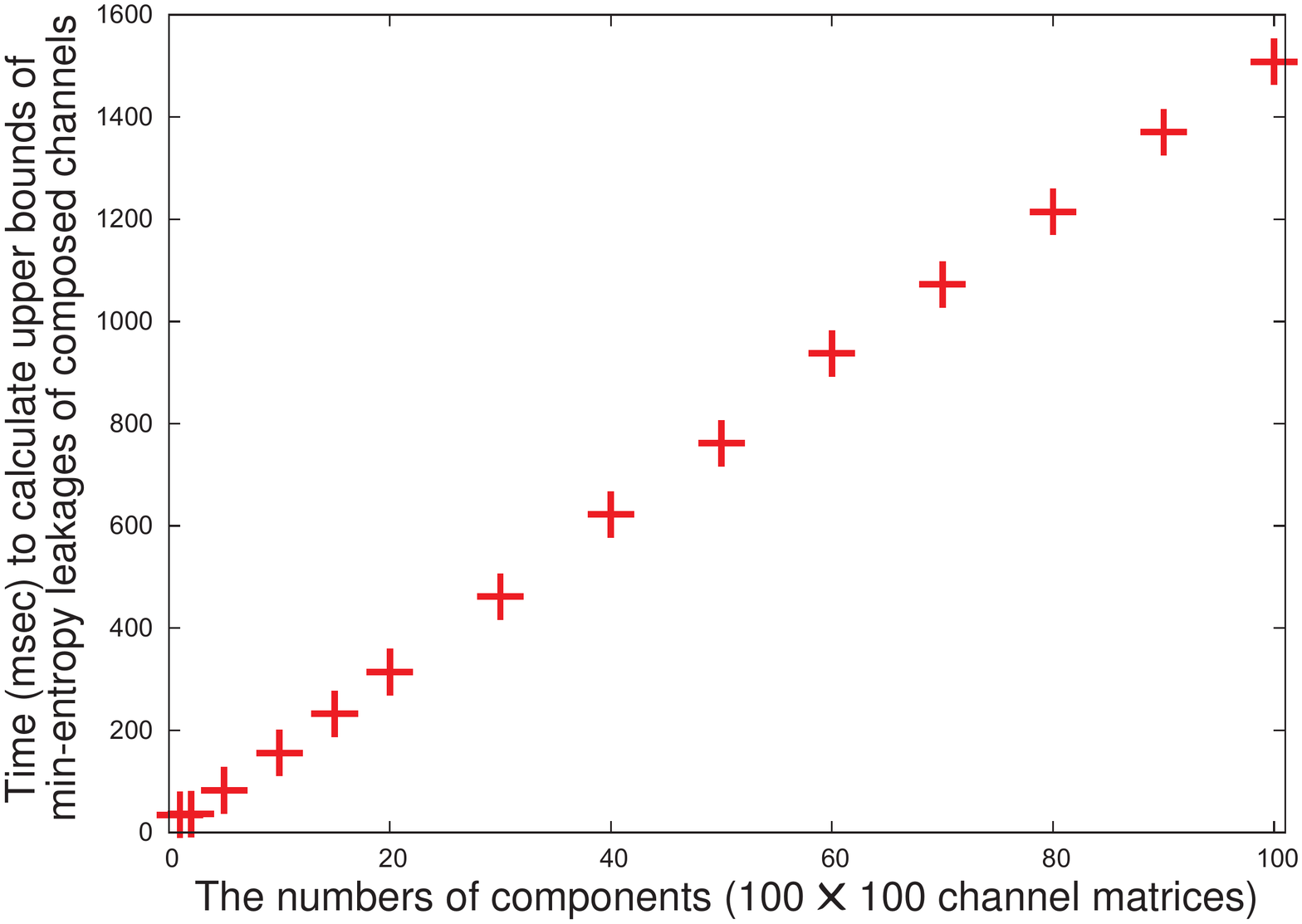}
  \caption{Average time to calculate upper bounds as a function of the number  of components.}
  \label{fig:time-large}
  \vspace{-1.5ex}
  \end{center}
 \end{minipage}
\end{figure}

Finally we evaluate the efficiency of  our method. We consider here the min-entropy leakage.
Figure~\ref{fig:time-comparison}  shows the execution time on a laptop (1.8 GHz Intel Core i5) for \leakiEst{} to compute the exact
min-entropy leakages of the channels composed of randomly generated
$10 \times 10$ component channels, in comparison with the time to compute
their upper bounds. To compute the exact leakages, we used \leakiEst{} with an option
that calculates the leakages from exact matrices.
As we can see, the execution time for the exact values increases rapidly. 
In fact, the size of composed channel increases exponentially with the number 
of components, so the complexity of this computation is at least exponential. 

For a large number of components, the time  to calculate upper bounds increases linearly as shown in Figure~\ref{fig:time-large}. 
As for the computation of the exact values with \leakiEst{}, we expected an exponential blow-up, but we could not check it since we run out of memory because of the size of the matrices.

%
\section{Related Work}
\label{sec:related}

Compositionality of \emph{qualitative} information flow properties has been widely studied e.g. in~\cite{Mantel02:SP,McLean94:SP,ZakinthinosL97:SP}.
Although security properties do not compose in general~\cite{McCullough88}, most type systems for secure information flow are compositional, which facilitates qualitative analyses of programs~\cite{SabelfeldM03}.

On the other hand, compositionality of \emph{quantitative} information flow is a more recent topic and not yet widely explored.
In~\cite{Boreale09:iandc}, Boreale shows the compositionality of the leakage and that of the leakage rate for the  Shannon entropy in a process calculus.
In~\cite{barthe2011information}, Barthe and K\"opf present the compositionality of the min-capacity for the parallel composition with distinct inputs.
They also consider the min-capacity of an adaptive composition $C_1+C_2$, which they call sequential composition, where the second channel $C_2$ receives both input and output of $C_1$. 
In~\cite{espinoza2013min}, Espinoza and Smith show the compositionality of min-entropy leakage and min-capacity for the parallel composition with shared input as well as for the cascade composition.
We are not aware of any prior compositionality result for $g$-leakage.

Recently more fine-grained composition has been considered. Kawamoto and Given-Wilson~\cite{KawamotoGW15qapl} show that our compositionality result gives an upper bound on the information leakage in scheduler-dependent systems where a scheduler interleaves the outputs of different components.

%
\section{Conclusion and Future Work}
\label{sec:conclude}

We have investigated  compositional methods to derive bounds on $g$-leakage.
To improve the precision of the bounds, we have proposed a technique based on the idea of approximating priors by 
removing small probabilities up to a parameter $\epsilon$.
From our experimental results we have found that the dependency of the precision on $\epsilon$ is not straightforward.
We leave for future work the problem of determining optimal values for $\epsilon$.
We also want to explore a possible relation between our technique and the notion of smooth entropies from the information theory literature~\cite{Cachin97}.
This could allow us to develop a more principled approach to the input approximation technique.

\section*{Acknowledgment}
\noindent This work has been partially supported by the project ANR-12-IS02-001 PACE, 
by the INRIA Equipe Associ\'{e}e PRINCESS, by the INRIA Large Scale Initiative CAPPRIS, and by  EU grant agreement no. 295261 (MEALS).
The work of Yusuke Kawamoto has been partially supported by a postdoc grant funded by the IDEX Digital Society project.
His work has also been partially supported by JSPS KAKENHI Grant Number JP17K12667 and by JSPS and Inria under the Japan-France AYAME Program.

\bibliographystyle{alpha} 
\bibliography{bib,short}

\appendix
\section{Omitted Proofs from Section~\ref{sec:compo-g-leak}}
\label{subsec:proofs}
In this section we present omitted proofs.
For brevity let $\X = \X_1 \times \X_2$, $\Y = \Y_1 \times \Y_2$ and $\W = \W_1 \times \W_2$.

%

%
%
The support $\S_{w_1, w_2}$, defined in Definition~\ref{def:full-support-set}, satisfies the following:
\begin{lem} \label{lem:jointly-supported}
Let $\pi$ be  jointly supported, and $g_1:\W_1  \times \X_1\rightarrow [0 ,1], g_2:\W_2  \times \X_2\rightarrow [0 ,1]$.
Then for all $(x_1, x_2) \in \X \setminus \S_{w_1, w_2}$,\, $\pi_1[x_1] \cdot \allowbreak g_1(w_1, x_1) \cdot \pi_2[x_2] \cdot g_2(w_2, x_2) = 0$.
\end{lem}
%
\proof
Let $(x_1, x_2) \in \X \setminus \S_{w_1, w_2}$.
By definition, we have $\pi[x_1, x_2] \cdot g((w_1, w_2), (x_1, x_2)) = 0$.
Since $\pi$ is jointly supported, $\pi_1[x_1] \cdot \pi_2[x_2] = 0$.
By the definition of joint gain functions,\, $g_1(w_1, x_1) g_2(w_2, x_2) = 0$.
Hence $\pi_1[x_1] g_1(w_1, x_1) \cdot \pi_2[x_2]  g_2(w_2, x_2) = 0$.
\qed

We present the proofs of Propositions~\ref{lem:disjoint-dep-g-prior},~\ref{lem:disjoint-dep-g-prior-upper} and Lemma~\ref{lem:disjoint-dep-g-vul} as follows.

\begin{proof}[Proof of Proposition~\ref{lem:disjoint-dep-g-prior}]
\[
\begin{array}{rl}
&
V_g(\pi) \cdot \Mminpi
\\[1.5ex] =\!&
\displaystyle
\hspace{-0.5ex}
\max_{(w_1, w_2) \in \W}
\hspace{-2.5ex}
\sum_{\hspace{2ex}(x_1, x_2) \in \X}
\hspace{-2.5ex}
\pi[x_1, x_2] \cdot g((w_1, w_2), (x_1, x_2))
\\[1.5ex] \!&
\displaystyle
\cdot
\min_{(w_1, w_2) \in \W}~ \min_{(x_1, x_2) \in \S_{w_1, w_2}} \frac{ \pi_1[x_1] g_1(w_1, x_1) \cdot \pi_2[x_2]  g_2(w_2, x_2) }{ \pi[x_1, x_2] \cdot g((w_1, w_2), (x_1, x_2)) }
\\[3.5ex] \le\!&
\displaystyle
\hspace{-0.5ex}
\max_{(w_1, w_2) \in \W}
\biggl(\,
\sum_{(x_1, x_2) \in \X}
\hspace{-1.5ex}
\pi[x_1, x_2] \cdot g((w_1, w_2), (x_1, x_2))
\\[0ex] \!&
\displaystyle
~\hspace{11ex}\cdot\hspace{-1.5ex}
~ \min_{(x_1, x_2) \in \S_{w_1, w_2}} \frac{ \pi_1[x_1] g_1(w_1, x_1) \cdot \pi_2[x_2]  g_2(w_2, x_2) }{ \pi[x_1, x_2] \cdot g((w_1, w_2), (x_1, x_2)) }
\biggr)
\\[3.5ex] \le\!&
\displaystyle
\hspace{-1.5ex}
\max_{(w_1, w_2) \in \W}
\hspace{-4.5ex}
\sum_{\hspace{4ex}(x_1, x_2) \in \S_{w_1, w_2}}
\hspace{-5.5ex}
\pi[x_1, x_2] g((w_1, w_2), (x_1, x_2)) \cdot
\frac{ \pi_1[x_1] g_1(w_1, x_1) \pi_2[x_2]  g_2(w_2, x_2) }{ \pi[x_1, x_2] g((w_1, w_2), (x_1, x_2)) }
\\[4.5ex] &
\hfill\mbox{($\because$ For all $(x_1, x_2) \in \X \setminus \S_{w_1, w_2}$,\, $\pi[x_1, x_2] g((w_1, w_2), (x_1, x_2)) = 0$.)}
\\[3ex] \le\!&
\displaystyle
\hspace{-1.5ex}
\max_{(w_1, w_2) \in \W}
\hspace{-0.0ex}
\sum_{(x_1, x_2) \in \X}
\hspace{-1.5ex}
\pi_1[x_1] g_1(w_1, x_1) \cdot \pi_2[x_2] g_2(w_2, x_2))
\\[4.5ex] =\!&
\displaystyle
\max_{w_1 \in \W_1}
\sum_{x_1 \in \X_1} \pi_1[x_1] \cdot g_1(w_1, x_1)
\cdot
\max_{w_2 \in \W_2}
\sum_{x_2 \in \X_2}  \pi_2[x_2] \cdot g_2(w_2, x_2)
\\[4.5ex] =\!&
V_{g_1}(\pi_1) \cdot V_{g_2}(\pi_2)
\texttt{.}
\end{array}
\]
Therefore 
\begin{align*}
\Hg(\pi)
= &
- \log V_g(\pi)
\\ \ge &
- \log V_{g_1}(\pi_1) - \log  V_{g_2}(\pi_2) + \log \Mminpi
\\ = &
H_{g_1}(\pi_1) + H_{g_2}(\pi_2) + \log \Mminpi
\texttt{.} \tag*{\qEd}
\end{align*}
\def\popQED{}
\end{proof}

%
\begin{proof}[Proof of Proposition~\ref{lem:disjoint-dep-g-prior-upper}]
\[
\begin{array}{rl}
&
V_g(\pi) \cdot \Mmaxpi
\\[1.5ex] =\!&
\displaystyle
\hspace{-0.5ex}
\max_{(w_1, w_2) \in \W}
\hspace{-1.0ex}
\sum_{(x_1, x_2) \in \X}
\hspace{-1.5ex}
\pi[x_1, x_2] \cdot g((w_1, w_2), (x_1, x_2))
\\[2.5ex] \!&
\displaystyle
\cdot
\max_{(w_1, w_2) \in \W}~ \max_{(x_1, x_2) \in \S_{w_1, w_2}}\hspace{-1.0ex} \frac{ \pi_1[x_1] g_1(w_1, x_1) \cdot \pi_2[x_2]  g_2(w_2, x_2) }{ \pi[x_1, x_2] \cdot g((w_1, w_2), (x_1, x_2)) }
\\[3.5ex] \ge\!&
\displaystyle
\hspace{-0.5ex}
\max_{(w_1, w_2) \in \W}
\Biggl(\,
\sum_{(x_1, x_2) \in \X}
\hspace{-1.5ex}
\pi[x_1, x_2] \cdot g((w_1, w_2), (x_1, x_2))
\\[-0.5ex] \!&
\displaystyle
~\hspace{11ex}\cdot\hspace{-1.5ex}
~ \max_{(x_1, x_2) \in \S_{w_1, w_2}}\hspace{-1.0ex} \frac{ \pi_1[x_1] g_1(w_1, x_1) \cdot \pi_2[x_2]  g_2(w_2, x_2) }{ \pi[x_1, x_2] \cdot g((w_1, w_2), (x_1, x_2)) }
\Biggr)
\\[3.5ex] \ge\!&
\displaystyle
\hspace{-0.5ex}
\max_{(w_1, w_2) \in \W}
\hspace{-2.5ex}
\sum_{\hspace{3ex}(x_1, x_2) \in \S_{w_1, w_2}}
\hspace{-4ex}
\pi_1[x_1]g_1(w_1, x_1) \cdot \pi_2[x_2] g_2(w_2, x_2)
\\[4ex] =\!&
\displaystyle
\hspace{-0.5ex}
\max_{(w_1, w_2) \in \W}
\hspace{-2.5ex}
\sum_{\hspace{2.5ex}(x_1, x_2) \in \X}
\hspace{-4ex}
\pi_1[x_1] g_1(w_1, x_1) \cdot \pi_2[x_2] g_2(w_2, x_2)
\hfill\mbox{($\because$ Lemma~\ref{lem:jointly-supported})}
\\[3.5ex] =\!&
\displaystyle
\biggl(\,
\max_{w_1 \in \W_1}
\sum_{x_1 \in \X_1} \pi_1[x_1] g_1(w_1, x_1)
\biggr)
\cdot
\biggl(\,
\max_{w_2 \in \W_2}
\sum_{x_2 \in \X_2}  \pi_2[x_2] g_2(w_2, x_2)
\biggr)
\\[3.5ex] =\!&
V_{g_1}(\pi_1) \cdot V_{g_2}(\pi_2)
\texttt{.}
\end{array}
\]
Therefore 
\begin{align*}
\Hg(\pi)
= &
- \log V_g(\pi)
\\ \le &
- \log V_{g_1}(\pi_1) - \log  V_{g_2}(\pi_2) + \log \Mmaxpi
\\ = &
H_{g_1}(\pi_1) + H_{g_2}(\pi_2) + \log \Mmaxpi
\texttt{.} \tag*{\qEd}
\end{align*}
\def\popQED{}
\end{proof}

%
%

\begin{proof}[Proof of Lemma~\ref{lem:disjoint-dep-g-vul}]
%
\[
\begin{array}{rl}
&
V_g(\pi, C_1\!\compd\!C_2) \cdot \Mminpi
\\[1.5ex] =\!&
\displaystyle
\hspace{-2ex}
\sum_{(y_1, y_2) \in \Y} 
\hspace{-0.5ex}
\max_{(w_1, w_2) \in \W}
\hspace{-3.5ex}
\sum_{~\hspace{3.0ex}(x_1, x_2) \in \X}
\hspace{-3.5ex}
\pi[x_1, x_2] \cdot (C_1\!\compd\!C_2)[(x_1, x_2), (y_1, y_2)] \cdot g((w_1, w_2), (x_1, x_2))
\\[3.5ex] \!&
\displaystyle
~\cdot\hspace{-0.5ex}
\min_{(w_1, w_2) \in \W}~ \min_{(x_1, x_2) \in \S_{w_1, w_2}} \frac{ \pi_1[x_1] g_1(w_1, x_1) \cdot \pi_2[x_2]  g_2(w_2, x_2) }{ \pi[x_1, x_2] \cdot g((w_1, w_2), (x_1, x_2)) }
\\[3.5ex] \le\!&
\displaystyle
\hspace{-2ex}
\sum_{(y_1, y_2) \in \Y} 
\hspace{-0.5ex}
\max_{(w_1, w_2) \in \W}
\biggl(\hspace{-3.5ex}
\sum_{~\hspace{3.0ex}(x_1, x_2) \in \X}
\hspace{-3.5ex}
\pi[x_1, x_2] \cdot (C_1\!\compd\!C_2)[(x_1, x_2), (y_1, y_2)] \cdot g((w_1, w_2), (x_1, x_2))
\\[3ex] \!&
\displaystyle
~\hspace{16ex}\cdot\hspace{-1.5ex}
~ \min_{(x_1, x_2) \in \S_{w_1, w_2}} \frac{ \pi_1[x_1] g_1(w_1, x_1) \cdot \pi_2[x_2]  g_2(w_2, x_2) }{ \pi[x_1, x_2] \cdot g((w_1, w_2), (x_1, x_2)) }
\biggr)
\\[3.5ex] \le\!&
\displaystyle
\hspace{-2ex}
\sum_{(y_1, y_2) \in \Y} 
\hspace{-0.5ex}
\max_{(w_1, w_2) \in \W}
\hspace{-6.5ex}
\sum_{~\hspace{5.0ex}(x_1, x_2) \in \S_{w_1, w_2}}
\hspace{-7.5ex}
\pi_1[x_1] g_1(w_1, x_1) \cdot \pi_2[x_2] g_2(w_2, x_2) \cdot (C_1\!\compd\!C_2)[(x_1, x_2), (y_1, y_2)]
\\[3.5ex] &
\hfill\mbox{($\because$ For all $(x_1, x_2) \in \X \setminus \S_{w_1, w_2}$,\, $\pi[x_1, x_2] g((w_1, w_2), (x_1, x_2)) = 0$.)}
\\[1.5ex] \le\!&
\displaystyle
\hspace{-2ex}
\sum_{(y_1, y_2) \in \Y} 
\hspace{-0.5ex}
\max_{(w_1, w_2) \in \W}
\hspace{-3.5ex}
\sum_{~\hspace{3.0ex}(x_1, x_2) \in \X}
\hspace{-3.5ex}
\pi_1[x_1] g_1(w_1, x_1) \cdot \pi_2[x_2] g_2(w_2, x_2) \cdot (C_1\!\compd\!C_2)[(x_1, x_2), (y_1, y_2)]
\\[3.5ex] =\!&
\displaystyle
\hspace{-2ex}
\sum_{(y_1, y_2) \in \Y} 
\max_{(w_1, w_2) \in \W}\hspace{-0.2ex}
\biggl(\,
\sum_{x_1 \in \X_1}\hspace{-1.0ex}
\pi_1[x_1] C_1[x_1, y_1]\, g_1(w_1, x_1)
\cdot\hspace{-0.5ex}
\sum_{x_2 \in \X_2}\hspace{-1.0ex}
\pi_2[x_2] C_2[x_2, y_2]\, g_2(w_2, x_2)%
\!\biggr)
\\[3.5ex] =\!&
\displaystyle
\hspace{-1ex}
\Bigl(\,
\sum_{y_1 \in \Y_1}\hspace{-1.0ex}
\max_{w_1 \in \W_1}\hspace{-1.0ex}
\sum_{x_1 \in \X_1}\hspace{-1.0ex}
\pi_1[x_1] C_1[x_1, y_1]\, g_1(w_1, x_1)%
\!\Bigr)%
\!\cdot\!%
\Bigl(\,
\sum_{y_2 \in \Y_2}\hspace{-1.0ex}
\max_{w_2 \in \W_2}\hspace{-1.0ex}
\sum_{x_2 \in \X_2}\hspace{-1.0ex}
\pi_2[x_2] C_2[x_2, y_2]\, g_2(w_2, x_2)%
\!\Bigr)
\\[3.5ex] =\!&
V_{g_1}(\pi_1, C_1) \cdot V_{g_2}(\pi_2, C_2)
\texttt{.}
\end{array}
\]
Therefore 
\begin{align*}
\Hg(\pi, C_1 \compd C_2)
 = &
- \log V_g(\pi, C_1 \compd C_2)
\\ \ge &
- \log V_{g_1}(\pi_1, C_1) - \log  V_{g_2}(\pi_2, C_2) + \log \Mminpi
\\ = &
H_{g_1}(\pi_1, C_1) + H_{g_2}(\pi_2, C_2) + \log \Mminpi
\texttt{.}
\end{align*}

%
\[
\begin{array}{rl}
&
V_g(\pi, C_1\!\compd\!C_2) \cdot \Mmaxpi
\\[1.5ex] =\!&
\displaystyle
\hspace{-2ex}
\sum_{(y_1, y_2) \in \Y} 
\hspace{-0.5ex}
\max_{(w_1, w_2) \in \W}
\hspace{-3.0ex}
\sum_{~\hspace{3.0ex}(x_1, x_2) \in \X}
\hspace{-3.5ex}
\pi[x_1, x_2] \cdot (C_1\!\compd\!C_2)[(x_1, x_2), (y_1, y_2)] \cdot g((w_1, w_2), (x_1, x_2))
\\[3.5ex] \!&
\displaystyle
~~~~~~~~~~~~~~~~~~~\cdot
\max_{(w_1, w_2) \in \W} \max_{(x_1, x_2) \in \S_{w_1, w_2}} \hspace{-1.5ex}\frac{ \pi_1[x_1] g_1(w_1, x_1) \cdot \pi_2[x_2]  g_2(w_2, x_2) }{ \pi[x_1, x_2] \cdot g((w_1, w_2), (x_1, x_2)) }
\\[3.5ex] \ge\!&
\displaystyle
\hspace{-2.5ex}
\sum_{(y_1, y_2) \in \Y} 
\hspace{-0.5ex}
\max_{(w_1, w_2) \in \W}
\biggl(\hspace{-3.0ex}
\sum_{~\hspace{3.0ex}(x_1, x_2) \in \X}
\hspace{-3.5ex}
\pi[x_1, x_2] \cdot (C_1\!\compd\!C_2)[(x_1, x_2), (y_1, y_2)] \cdot g((w_1, w_2), (x_1, x_2))%
\\[3.5ex] \!&
\displaystyle
~\hspace{15ex}\cdot\hspace{-1.5ex}
~ \max_{(x_1, x_2) \in \S_{w_1, w_2}}
\hspace{-2.0ex}\frac{ \pi_1[x_1] g_1(w_1, x_1) \cdot \pi_2[x_2]  g_2(w_2, x_2) }{ \pi[x_1, x_2] \cdot g((w_1, w_2), (x_1, x_2)) }%
\!\biggr)
\\[3.5ex] \ge\!&
\displaystyle
\hspace{-2ex}
\sum_{(y_1, y_2) \in \Y} 
\hspace{-0.5ex}
\max_{(w_1, w_2) \in \W}
\hspace{-7ex}
\sum_{~\hspace{6ex}(x_1, x_2) \in \S_{w_1, w_2}}
\hspace{-7ex}
\pi_1[x_1] g_1(w_1, x_1) \cdot \pi_2[x_2] g_2(w_2, x_2) \cdot (C_1\!\compd\!C_2)[(x_1, x_2), (y_1, y_2)]
\\[3.5ex] =\!&
\displaystyle
\hspace{-2ex}
\sum_{(y_1, y_2) \in \Y} 
\hspace{-0.5ex}
\max_{(w_1, w_2) \in \W}
\hspace{-4.0ex}
\sum_{~\hspace{4.0ex}(x_1, x_2) \in \X}
\hspace{-4ex}
\pi_1[x_1] g_1(w_1, x_1) \cdot \pi_2[x_2] g_2(w_2, x_2) \cdot (C_1\!\compd\!C_2)[(x_1, x_2), (y_1, y_2)]
\\[-1.5ex] &
\hfill\mbox{($\because$ Lemma~\ref{lem:jointly-supported})}
\\[1ex] =\!&
\displaystyle
\hspace{-1ex}
\Bigl(\,
\sum_{y_1 \in \Y_1}\hspace{-0.5ex}
\max_{w_1 \in \W_1}\hspace{-1.0ex}
\sum_{x_1 \in \X_1}\hspace{-1.0ex}
\pi_1[x_1] C_1[x_1, y_1]\, g_1(w_1, x_1)%
\!\Bigr)
\cdot
\Bigl(\,
\sum_{y_2 \in \Y_2}\hspace{-0.5ex}
\max_{w_2 \in \W_2}\hspace{-1.0ex}
\sum_{x_2 \in \X_2}\hspace{-1.0ex}
\pi_2[x_2] C_2[x_2, y_2]\, g_2(w_2, x_2)%
\!\Bigr)
\\[3.5ex] =\!&
V_{g_1}(\pi_1, C_1) \cdot V_{g_2}(\pi_2, C_2)
\texttt{.}
\end{array}
\]
Therefore 
\begin{align*}
\Hg(\pi, C_1 \compd C_2)
= &
- \log V_g(\pi, C_1 \compd C_2)
\\ \le &
- \log V_{g_1}(\pi_1, C_1) - \log  V_{g_2}(\pi_2, C_2) + \log \Mmaxpi
\\ = &
H_{g_1}(\pi_1, C_1) + H_{g_2}(\pi_2, C_2) + \log \Mmaxpi
\texttt{.} \tag*{\qEd}
\end{align*}
\def\popQED{}
\end{proof}

\end{document}